\documentclass[a4paper,onecolumn,11pt,unpublished]{quantumarticle}
\pdfoutput=1
\usepackage{envLatex}

\begin{document}

\title{Continuous-variable nonlocality and contextuality
}

\author{Rui Soares Barbosa}
\email{rui.soaresbarbosa@inl.int}
\orcid{0000-0002-0465-8518}
\affiliation{INL -- International Iberian Nanotechnology Laboratory, Braga, Portugal}
\author{Tom Douce}
\email{doucetom@hotmail.fr}
\affiliation{School of Informatics, University of Edinburgh, Edinburgh, United Kingdom}
\author{Pierre-Emmanuel Emeriau}
\email{pierre-emmanuel.emeriau@lip6.fr}
\orcid{0000-0001-5155-1783}
\affiliation{Sorbonne Universit\'e, CNRS, LIP6, F-75005 Paris, France}
\author{Elham Kashefi}
\email{elham.kashefi@lip6.fr}
\affiliation{School of Informatics, University of Edinburgh, Edinburgh, United Kingdom}
\affiliation{Laboratoire d'Informatique de Paris 6, CNRS and Sorbonne Universit\'e, Paris, France}
\author{Shane Mansfield}
\email{shane.mansfield@quandela.com}
\affiliation{Quandela, 7 Rue L\'eonard de Vinci, Massy, France}
\thanks{\linebreak \linebreak \scriptsize The list of authors is ordered alphabetically.}

\maketitle

\hypersetup{urlcolor=magenta}

\begin{abstract}
Contextuality is a non-classical behaviour that can be exhibited by quantum systems.
It is increasingly studied for its relationship to quantum-over-classical advantages in informatic tasks.
To date, it has largely been studied in discrete-variable scenarios, where observables take values in discrete and usually finite sets. Practically, on the other hand, continuous-variable scenarios offer some of the most promising candidates for implementing quantum computations and informatic protocols.
Here we set out a framework for treating contextuality in continuous-variable scenarios.
It is shown that the Fine--Abramsky--Brandenburger theorem extends to this setting, an important consequence of which is that Bell nonlocality can be viewed as a special case of contextuality, as in the discrete case.
The contextual fraction, a quantifiable measure of contextuality that bears a precise relationship to Bell inequality violations and quantum advantages, is also defined in this setting.
It is shown to be a non-increasing monotone with respect to classical operations that include binning to discretise data.
Finally, we consider how the contextual fraction can be formulated as an infinite linear program. Through Lasserre relaxations, we are able to express this infinite linear program as a hierarchy of semi-definite programs that allow to calculate the contextual fraction with increasing accuracy.
\end{abstract}

    \section*{Introduction}\label{intro}
        Contextuality is one of the principal markers of non-classical behaviour that can be exhibited by quantum systems.
The Heisenberg uncertainty principle identified that certain pairs of quantum observables are incompatible, e.g.~position and momentum.
In operational terms, observing one will disturb the outcome statistics of the other.
This is sometimes cited as evidence that not all observables can simultaneously be assigned definite values.
Taking the mathematical formalism of quantum mechanics at face value, that is indeed the case, in stark contrast with classical physical theories.
However, one may wonder whether it is possible to build a (presumably more fundamental) theory more in accordance with our classical intuitions, but which still \emph{matches} the empirical predictions of quantum mechanics.
Put briefly, the fundamental question is then
whether such quantum oddities
are a necessary property of any theory that accurately describes nature, and thus have \emph{empirical content},
or mere artifices of the mathematical formalism of quantum theory.

This question might be answered by attempting to build a hidden-variable model
reproducing quantum-mechanical empirical predictions but with the further assumption that it be \textit{noncontextual}.
Roughly speaking, the latter imposes that the model must respect the basic assumptions that
\begin{enumerate*}[label=(\roman*)]
\item
hidden variables assign definite values to all the observable properties,
and
\item
jointly performing \stress{compatible} observables does not disturb the hidden variable.
\end{enumerate*}
That these apparently simple assumptions are at odds with the empirical predictions of quantum mechanics is the content of the seminal theorems by Bell \cite{bell1966} and by Kochen \& Specker \cite{ks}.

Separately to its foundational importance, contextuality also has a more practical significance.
A major application of quantum theory today is in quantum information and computation.
There, one is primarily interested in what can be \emph{done} with quantum systems and is beyond the capabilities of any classical implementation.
So one is interested in the properties of the \emph{correlations} realisable by quantum systems when compared to the kind of correlations that could arise from any classical theory.
In this sense, aside from whatever foundational or physical significance one may wish (or not) to ascribe to contextuality,
it has an undeniable practical significance in relation to quantum information and computation.
In particular, it has now been identified as the essential ingredient for enabling a range of quantum-over-classical advantages in informatic tasks, which include the onset of universal quantum computing in certain computational models \cite{raussendorf2013contextuality,howard2014contextuality,abramsky2017contextual,bermejo2017contextuality,abramsky2017quantum}.

It is notable that to date the study of contextuality has largely focused on discrete variable scenarios and that the main frameworks and approaches to contextuality are tailored to modelling these, e.g.~\cite{ab,csw,afls,dzhafarov2015contextuality}.
In such scenarios, observables can only take values in discrete, and usually finite, sets.
Discrete variable scenarios typically arise in finite-dimensional quantum mechanics, \eg when dealing with quantum registers in the form of systems of multiple qubits as is common in quantum information and computation theory.

Yet, from a practical perspective, continuous-variable quantum systems are emerging
as some of the most promising candidates for implementing
quantum informational and computational tasks \cite{braunstein2005quantum,weedbrook2012gaussian}.
The main reason for this is that they offer
unrivalled possibilities for deterministic generation
of large-scale resource states \cite{yoshikawa2016invited}
and for highly-efficient measurements of certain observables.
Together these cover many of the basic operations required in the one-way or
measurement-based model of quantum computing \cite{raussendorf2001one}, for example.
Typical implementations are in optical systems where the continuous variables
correspond to the position-like and momentum-like quadratures of
the quantised modes of an electromagnetic field.
Indeed position and momentum,
as mentioned previously in relation to the uncertainty principle,
are the prototypical examples of continuous variables in quantum mechanics.

Since quantum mechanics itself is infinite dimensional,
it also makes sense from a foundational perspective to extend
analyses of the key concept of contextuality to the continuous-variable setting.
Furthermore, continuous variables can be useful when dealing with iteration,
even when attention is restricted to finite-variable actions at discrete time steps,
as is traditional in informatics.
An interesting question, for example, is whether contextuality arises and is of interest in such situations
as the infinite behaviour of quantum random walks.

The main contributions of this article are the following:
\begin{itemize}
\item
We present a robust framework for contextuality in continuous-variable scenarios that follows along the lines of the discrete-variable framework introduced by Abramsky and Brandenburger \cite{ab} (Section~\ref{sec:formalism}). 
We thus generalise this framework to deal with outcomes being valued on general measurable spaces, as well as to arbitrary (infinite) sets of measurement labels.

\item
We show that the Fine--Abramsky--Brandenburger theorem \cite{fine1982hidden,ab} extends to continuous variables (Section~\ref{sec:fine}).
This establishes that noncontextuality of an empirical behaviour,
originally characterised by the existence of a deterministic hidden-variable model \cite{bell1966,ks}, can equivalently be characterised by the existence of a factorisable hidden-variable model, and that ultimately both of these are subsumed by a canonical form of hidden-variable model -- a \stress{global section} in the sheaf-theoretic perspective.
An important consequence is that Bell nonlocality may be viewed as a special case of contextuality in continuous-variable scenarios just as for discrete-variable scenarios.
\item
The contextual fraction, a quantifiable measure of contextuality that bears a precise relationship to Bell inequality violations and quantum advantages \cite{abramsky2017contextual}, can also be defined in this setting using infinite linear programming (Section~\ref{sec:quantifying}).
It is shown to be a non-increasing monotone with respect to the free operations of a resource theory for contextuality \cite{abramsky2017contextual,abramsky2019comonadic}.
Crucially, these include the common operation of binning to discretise data.
A consequence is that any witness of contextuality on discretised empirical data also witnesses and gives a lower bound on genuine continuous-variable contextuality.
\item
While the infinite linear programs are of theoretical importance
and capture exactly the quantity and Bell-like inequalities in which we are interested,
they are not directly useful for actual numerical computations.
To get around this limitation, we introduce a hierarchy of semi-definite
programs which are relaxations of the original problem and whose values
converge monotonically to the contextual fraction (Section~\ref{sec:sdp}). This applies in the restricted setting where there is a finite set of measurement labels.
\end{itemize}

\paragraph{Related work.}
Note that we are specifically interested in scenarios involving observables with continuous spectra, or in more operational language, measurements with continuous outcome spaces.
We still consider scenarios featuring only discrete sets of observables or measurements, as is typical in continuous-variable quantum computing.
The possibility of considering contextuality in settings with continuous measurement spaces has also been evoked in \cite{cunha2019measures}.
We also note that several prior works have explicitly considered contextuality in continuous-variable systems \cite{plastino2010state,he2010bell,mckeown2011testing,su2012quantum,asadian2015contextuality,laversanne2017general,ketterer2018continuous}.
Our approach is distinct from these in that it provides a genuinely continuous-variable treatment of contextuality itself as opposed to embedding discrete variable contextuality arguments into, or extracting them from, continuous-variable systems.

    \section{Continuous-variable behaviours}\label{sec:cvbehaviours}
        In this section we provide a brief motivational example for the kind of continuous-variable empirical behaviour we are interested in analysing. 
The approach applies generally to any hypothetical empirical data, including those that do not admit a quantum realisation (e.g.\ the PR box from Ref~\cite{popescu1994quantum}). But also, in particular, it does of course apply to empirical data arising from quantum mechanics, in that the statistics arise from a state and measurements on a quantum system according to the Born rule.
Indeed, quantum theory provides the main motivation for this study and more broadly for the sheaf-theoretic approach, because of a feature that may arise in empirical models having quantum but not classical realisations: which we refer to as \emph{contextuality}.

Suppose that we can interact with a system by performing measurements on it and observing their outcomes.
A feature of quantum systems is that not all observables commute, so that certain combinations of measurements are incompatible.

At best, we can obtain empirical observational data for contexts in which only compatible measurements are performed, which can be collected by running the experiment repeatedly.
As we shall make more precise in  Sections~\ref{sec:formalism} and \ref{sec:fine}, contextuality arises when the empirical data obtained is inconsistent with the assumption that for each run of the experiment the system has a global and context-independent assignment of values to all of its observable properties.

To take an operational perspective, a typical example of an experimental setup or scenario that we consider is the one depicted in Figure~\ref{fig:scenario} [left].
In this scenario, a system is prepared in some fixed bipartite state, following which parties $A$ and $B$ may each choose between two measurement settings, $m_A \in \enset{ a, a'}$ for $A$ and $m_B \in \enset{ b, b'}$ for $B$.
We assume that outcomes of each measurement live in $\bfR$, which typically will be a bounded measurable subspace of the real numbers $\R$ (with its Borel $\sigma$-algebra).
Depending on which choices of inputs were made, the empirical data might for example be distributed according to one of the four hypothetical probability density plots in $\bfR^2$ depicted in Figure~\ref{fig:scenario} [right].
This scenario and hypothetical empirical behaviour has been considered elsewhere \cite{ketterer2018continuous} as a continuous-variable version of the Popescu--Rohrlich (PR) box \cite{popescu1994quantum}.

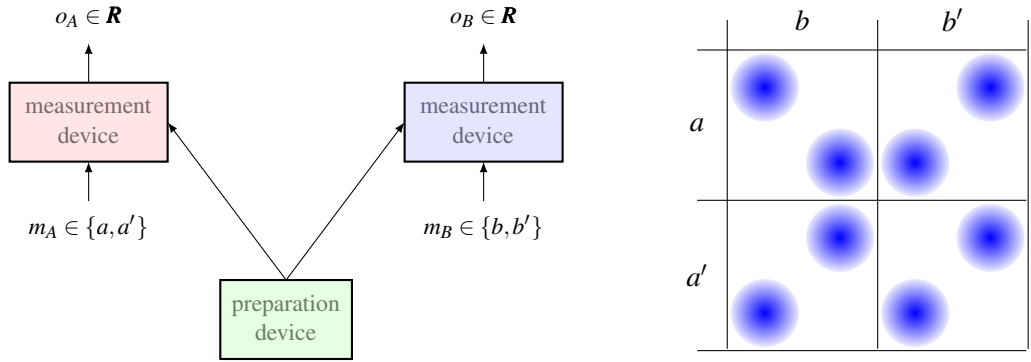
\begin{figure}[tbhp]
\centering
{\footnotesize
    \begin{tikzpicture}[scale=3]
  \path (0.13,0) coordinate (A1);
  \path (0.13,0.3) coordinate (A2);
  \path (0.87,0.3) coordinate (A3);
  \path (0.87,0) coordinate (A4);
  \path (0.5,-0.15) coordinate (B1);
  \path (0.5,0) coordinate (B2);
  \path (0.5,0.3) coordinate (B3);
  \path (0.5,0.45) coordinate (B4);
  \path (0.5,0.15) node {\begin{tabular}{c} measurement \\ device \end{tabular}};
  \path (0.5,-0.25) node {$m_A \in \{a,a'\}$};
  \path (0.5,0.55) node {$o_A \in \bfR$};
  \begin{scope}[ thick,black]
    \filldraw[fill=red!20,fill opacity=0.5] (A1) -- (A2) -- (A3) -- (A4) -- cycle;
  \end{scope} 
  \draw[->] (B1) to (B2);  
  \draw[->] (B3) to (B4);
    
  \path (1.63,0) coordinate (a1);
  \path (1.63,0.3) coordinate (a2);
  \path (2.37,0.3) coordinate (a3);
  \path (2.37,0) coordinate (a4);
  \path (2.,-0.15) coordinate (b1);
  \path (2.,0) coordinate (b2);
  \path (2.,0.3) coordinate (b3);
  \path (2.,0.45) coordinate (b4);
  \path (2.,0.15) node {\begin{tabular}{c} measurement \\ device \end{tabular}};
  \path (2.,-0.25) node {$m_B \in \{b,b'\}$};
  \path (2.,0.55) node {$o_B \in \bfR$};
  \begin{scope}[ thick,black]
    \filldraw[fill=blue!20,fill opacity=0.5] (a1) -- (a2) -- (a3) -- (a4) -- cycle;
  \end{scope} 
  \draw[->] (b1) to (b2);  
  \draw[->] (b3) to (b4);
  
  \path (.93,-.75) coordinate (s1);
  \path (.93,-.45) coordinate (s2);
  \path (1.57,-.45) coordinate (s3);
  \path (1.57,-.75) coordinate (s4);
  \path (1.25,-.6) node {\begin{tabular}{c} preparation \\ device \end{tabular}};
  \begin{scope}[ thick,black]
    \filldraw[fill=green!20,fill opacity=0.5] (s1) -- (s2) -- (s3) -- (s4) -- cycle;
  \end{scope} 
  
  \path (1.25,-.45) coordinate (p);
  \path (0.87,0.15) coordinate (ma);
  \path (1.63,0.15) coordinate (mb);
  
  \draw[->] (p) to (ma);
  \draw[->] (p) to (mb);
  
\end{tikzpicture}
    }
    \quad\quad\quad
\begin{tikzpicture}[scale=.3]

\def\rad{1.15}

\node [inner sep=.0pt] (a11) at (0,0) {};
\node [inner sep=.0pt] (a12) at (5,0) {};
\node [inner sep=.0pt] (a13) at (10,0) {};
\node [inner sep=.0pt] (a21) at (0,5) {};
\node [inner sep=.0pt] (a22) at (5,5) {};
\node [inner sep=.0pt] (a23) at (10,5) {};
\node [inner sep=.0pt] (a31) at (0,10) {};
\node [inner sep=.0pt] (a32) at (5,10) {};
\node [inner sep=.0pt] (a33) at (10,10) {};

\node (left) at (-1,0) {};
\node (up) at (0,1) {};

\node (A') at ($.5*(a11)+.5*(a21)+(left)$) {$a'$};
\node (A) at ($.5*(a21)+.5*(a31)+(left)$) {$a$};
\node (B) at ($.5*(a31)+.5*(a32)+(up)$) {$b$};
\node (B') at ($.5*(a32)+.5*(a33)+(up)$) {$b'$};

\draw ($(a11)+(left)$) -- (a13);
\draw ($(a21)+(left)$) -- (a23);
\draw ($(a31)+(left)$) -- (a33);

\draw ($(a31)+(up)$) -- (a11);
\draw ($(a32)+(up)$) -- (a12);
\draw ($(a33)+(up)$) -- (a13);

\filldraw[white,inner color=blue,outer color=blue!5] ($.75*(a11)+.25*(a22)$) circle (\rad);
\filldraw[white,inner color=blue,outer color=blue!5] ($.25*(a11)+.75*(a22)$) circle (\rad);

\filldraw[white,inner color=blue,outer color=blue!5] ($.75*(a12)+.25*(a23)$) circle (\rad);
\filldraw[white,inner color=blue,outer color=blue!5] ($.25*(a12)+.75*(a23)$) circle (\rad);

\filldraw[white,inner color=blue,outer color=blue!5] ($.75*(a22)+.25*(a33)$) circle (\rad);
\filldraw[white,inner color=blue,outer color=blue!5] ($.25*(a22)+.75*(a33)$) circle (\rad);

\filldraw[white,inner color=blue,outer color=blue!5] ($.75*(a22)+.25*(a31)$) circle (\rad);
\filldraw[white,inner color=blue,outer color=blue!5] ($.25*(a22)+.75*(a31)$) circle (\rad);

\end{tikzpicture}
\label{fig:scenario}
\caption{
[Left] operational depiction of a typical bipartite experimental scenario. [Right] Hypothetical probability density plots for empirical data arising from such an experiment. Cf.~the discrete-variable probability tables of \cite{mansfield2012hardy,mansfield2017consequences}.
}
\end{figure}

    \section{Preliminaries on measures and probability}\label{sec:preliminaries}
        In order to properly treat probability on continuous-variable spaces, it is necessary to introduce a modicum of measure theory.
This section serves to recall some basic ideas and to fix notation.
The reader may choose to skip the section
and consult it as reference for the remainder of the article.

A \emph{measurable space} is a pair $\bfX = \tuple{X,\Fc}$ consisting of a set $X$ and a $\sigma$-algebra (or $\sigma$-field) $\Fc$ on $X$,
\ie a family of subsets of $X$
containing the empty set and closed under complementation and countable unions.
In some sense, this specifies the subsets of $X$ that can be assigned a `size', and which are therefore called the \emph{measurable sets} of $\bfX$.
Throughout this paper, we follow the convention of using boldface to denote the measurable space and the same symbol in regular face for its underlying set.

A trivial example of a $\sigma$-algebra over any set $X$ is its powerset $\ps(X)$, which gives the discrete measurable space $\tuple{X,\ps(X)}$, where every set is measurable.
This is typically used when $X$ is countable (finite or countably infinite), in which case this discrete $\sigma$-algebra 
is generated by the singletons.
Another example, of central importance in measure theory, is $\tuple{\R,\Bc_\R}$, where $\Bc_\R$ is the $\sigma$-algebra generated from the open sets of $\R$, whose elements are called the Borel sets.
Working with Borel sets avoids the problems that would arise if we naively attempted to measure or assign probabilities to points in the continuum.
More generally, any topological space gives rise to a Borel measurable space in this fashion.

A \emph{measurable function} between measurable spaces $\bfX = \tuple{X,\Fc_X}$
and $\bfY = \tuple{Y, \Fc_Y}$
is a function $\fdec{f}{X}{Y}$ between the underlying sets whose preimage preserves measurable sets, \ie such that, for any $E \in \Fc_Y$,
${f^{-1}(E) \in \Fc_X}$.
This is analogous to the definition of a continuous function between topological spaces.
Clearly, the identity function is measurable and measurable functions compose. We will denote by $\catMeas$ the category whose objects are measurable spaces 
and whose morphisms are measurable functions.

The \emph{product} of two measurable spaces $\bfX_1 = \tuple{X_1,\Fc_1}$ and $\bfX_2 = \tuple{X_2,\Fc_2}$ is
the measurable space
\begin{equation*}
\bfX_1 \times \bfX_2 = \tuple{X_1 \times X_2, \Fc_1 \otimes \Fc_2} \Mcomma
\end{equation*}
where the Cartesian product of the underlying sets, $X_1 \times X_2$, is equipped with the so-called tensor product $\sigma$-algebra $\Fc_1 \otimes \Fc_2$, which is the $\sigma$-algebra generated by the `rectangles', subsets of the form $E_1 \times E_2$ with $E_1 \in \Fc_1$ and $E_2 \in \Fc_2$.
This is the categorical (binary) product in $\catMeas$.

We shall also need to deal with \emph{infinite products} of measurable spaces.
The generalisation is analogous to that for products of topological spaces, where the box topology (generated by `rectangles') is no longer the most natural choice when dealing with infinite families, but rather the topology generated by `cylinders'.
Let $I$ be an arbitrary index set. 
The product of measurable spaces $(\bfX_i = \tuple{X_i,\Fc_i})_{i \in I}$ is the measurable space
\begin{equation}
\prod_{i \in I} \bfX_i = \tuple{ \prod_{i \in I} X_i, \bigotimes_{i \in I} \Fc_i}  \Mcomma
\label{eq:infiniteproduct}
\end{equation}
where $X_I = \prod_{i \in I} X_i$ is the Cartesian product of the underlying sets, and $\Fc_I = \bigotimes_{i \in I} \Fc_i$ is the $\sigma$-algebra generated by subsets of $\prod_{i \in I} X_i$ of the form $\prod_{i \in I} E_i$ where $E_i \subseteq X_i$ for all $i \in I$ 
and $E_i \neq X_i$ for only finitely many $i \in I$.
This is the smallest $\sigma$-algebra that makes the projection maps $\fdec{\pi_k}{\prod_{i \in I} X_i }{X_k}$ measurable. It therefore corresponds to the categorical (arbitrary) product in $\catMeas$.

A \emph{measure} on a measurable space $\bfX = \tuple{X,\Fc}$ is a function $\fdec{\mu}{\Fc}{\Rext}$
from the $\sigma$-algebra to the extended real numbers $\Rext = \R \cup \enset{-\infty,+\infty}$
satisfying:
\begin{enumerate}[label=(\roman*)]
\item\label{cond:nonneg}

[nonnegativity]
$\mu(E)\geq 0$ for all $E\in\Fc$;
\item\label{cond:null}

[null empty set]
$\mu(\emptyset)=0$;
\item\label{cond:sigmaadd}

[$\sigma$-additivity]
for any countable family $\family{E_i}_{i=1}^\infty$ of pairwise disjoint measurable sets, it holds that
$\mu(\bigcup_{i=1}^\infty E_i) = \sum_{i=1}^\infty \mu(E_i)$.
\end{enumerate}

A measure on $\bfX$ allows one to \emph{integrate} well-behaved\footnotemark\ measurable functions $\fdec{f}{\bfX}{\tuple{\R,\Bc_\R}}$ to obtain a real value,
denoted $\intg{\bfX}{f}{\mu}$ or $\intg{x\in\bfX}{f(x)}{\mu(x)}$.
The simplest example of such a measurable function is the \emph{indicator function} of a measurable set $E \in \Fc$:
\[\chi_{_E}(x) \defeq \begin{cases} 1 & \text{if $x \in E$} \\ 0 & \text{if $x \not\in E$.}\end{cases}\]
For any measure $\mu$ on $\bfX$, its integral yields \begin{equation}\label{eq:integralindicator}\intg{\bfX}{\chi_{_E}}{\mu} = \mu(E) \Mdot \end{equation}

A measure $\mu$ is finite if $\mu(X)<\infty$ and in particular it is a \emph{probability measure} if $\mu(X)=1$.
We will denote by $\Measures(\bfX)$ and $\PMeasures(\bfX)$, respectively, the sets of measures and probability measures on the measurable space $\bfX$.

\footnotetext{For a comprehensive treatment, we refer the reader to \eg \cite{Bill79}, or to \cite{panangaden2009labelled} for a beautiful and more concise introduction aimed particularly at computer scientists.}

A measurable function $\fdec{f}{\bfX}{\bfY}$
carries any measure $\mu$ on $\bfX$ to
a measure $f_*\mu$ on $\bfY$. This \emph{push-forward} measure is given by
$f_*\mu(E) = \mu(f^{-1}(E))$ for any set $E$ measurable in $\bfY$.
An important use of push-forward measures is that for any integrable function $\fdec{g}{\bfY}{\tuple{\R,\Bc_\R}}$,
it allows us to write the following change-of-variables formula
\begin{equation}\label{eq:changeofvariables}
\intg{\bfY}{g}{f_*\mu} = \intg{\bfX}{g \circ f}{\mu}\Mdot
\end{equation}
The push-forward operation preserves the total measure, hence it takes $\PMeasures(\bfX)$ to $\PMeasures(\bfY)$.

A case that will be of particular interest to us
is the push-forward of a measure $\mu$ on a product space $\bfX_1 \times \bfX_2$ along a projection 
$\fdec{\pi_i}{\bfX_1 \times \bfX_2}{\bfX_i}$: this yields
the \emph{marginal measure} $\mu|_{\bfX_i}={\pi_i}_*\mu$,
where e.g. for $E$ measurable in $\bfX_1$, 
$\mu|_{\bfX_1}(E) = \mu(\pi_1^{-1}(E)) = \mu(E \times X_2)$.

In the opposite direction, given measures $\mu_1$ on $\bfX_1$ and $\mu_2$ on  $\bfX_2$,
a \emph{product measure} $\mu_1 \times \mu_2$ is a measure on the product measurable space $\bfX_1 \times \bfX_2$ satisfying
$(\mu_1 \times \mu_2)(E_1 \times E_2) = \mu_1(E_1)\mu_2(E_2)$
for all $E_1 \in \Fc_1$ and $E_2 \in \Fc_2$.
For probability measures, there is a uniquely determined product measure.\footnotemark\
The analogous, much more general statement also holds for arbitrary products of probability measures (see \eg \cite[section 11.2]{vestrup2003measures}).

\footnotetext{In fact, this holds more generally for $\sigma$-finite measures,
\ie when the space is a countable union of sets of finite measure.
}

We can view $\Measures$ as a map that takes a measurable space to the set of measures on that space,
and similarly for $\PMeasures$.
These become functors $\catMeas \lto \catSet$ if we define the action on morphisms to be the push-forward operation. Explicitly we set $\fdecdef{\Measures(f) \defeq f_*}{\Measures(\bfX)}{\Measures(\bfY)}{\mu}{f_*\mu}$, where $\fdec{f}{\bfX}{\bfY}$ is a measurable function, and similarly for $\PMeasures$.

Remarkably, the set $\PMeasures(\bfX)$ of probability measures on $\bfX$ can itself be made into a measurable space by equipping it with the least $\sigma$-algebra that makes the evaluation functions 
\begin{equation*} 
\fdecdef{\mathsf{ev}_E}{\PMeasures(\bfX)}{[0,1]}{\mu}{\mu(E)} 
\end{equation*}
measurable for all $E \in \Fc_X$.\footnotemark
\footnotetext{More concretely, it is the $\sigma$-algebra generated by the sets $\mathsf{ev}_E^{-1}([0,r))=\setdef{\mu \in \PMeasures(X)}{\mu(E) < r}$ with  $E \in \Fc_X$ and $r \in [0,1]$.}
This yields an endofunctor
$\fdec{\PMeasures}{\catMeas}{\catMeas}$,
which moreover has the structure of a monad, called the Giry monad \cite{giry1982categorical}.
The unit of this monad is given by 
\begin{equation*}
\fdecdef{\eta_{\bfX}}{\bfX}{\PMeasures(\bfX)}{x}{\delta_x}
\end{equation*}
where $\delta_x$ is the \emph{Dirac measure}, or point mass, at $x$
given by $\delta_x(E) \defeq  \chi_{_E}(x)$.
Multiplication of the monad is given by
\begin{equation*}
\fdec{\mu_\bfX}{\PMeasures(\PMeasures(\bfX))}{\PMeasures(\bfX)}
\end{equation*}
which takes a probability measure $P$ on $\PMeasures(\bfX)$ to its `average', 
a probability measure $\mu_\bfX(P)$ on $\bfX$, $\fdec{\mu_\bfX(P)}{\Fc_X}{[0,1]}$,  whose value on a measurable set $E \in \Fc_X$ is given by $\mu_\bfX(P)(E) \defeq \intg{\PMeasures(\bfX)}{\mathsf{ev}_E}{P}$.

The Kleisli category of this monad is the category of Markov kernels, which represent continuous-variable probabilistic maps and generalise the discrete notion of stochastic matrix.
Concretely, a \emph{Markov kernel} between measurable spaces $\bfX = \tuple{X,\Fc_X}$
and $\bfY = \tuple{Y, \Fc_Y}$ is a function
$\fdec{k}{X \times \Fc_Y}{[0,1]}$
such that:
\begin{enumerate}[label=(\roman*)]
    \item for all $E \in \Fc_Y$, $\fdec{k(\dummy,E)}{X}{[0,1]}$ is a measurable function;\footnote{The space $[0,1]$ is assumed to be equipped with its Borel $\sigma$-algebra.}
    \item for all $x \in X$, $\fdec{k(x,\dummy)}{\Fc_Y}{[0,1]}$ is a probability measure.
\end{enumerate}

    \section{Framework}\label{sec:formalism}
        In this section, we follow closely the discrete-variable framework of \cite{ab} in more formally describing the kinds of experimental scenarios in which we are interested and the empirical behaviours that arise on these, although some extra care is required for dealing with continuous variables.

\subsection{Measurement scenarios}

\begin{definition}
A \emph{measurement scenario} is a triple $\XMOO$ whose elements are specified as follows.
\begin{itemize}
    \item  $X$ is a (possibly  infinite) set of \emph{measurement labels}.
    
    \item $\Mc$ is a covering family of subsets of $X$, \ie such that $\bigcup\Mc = X$.
    The elements $C \in \Mc$ are called \emph{maximal contexts} and represent maximal sets of compatible observables. We therefore require that $\Mc$ be an anti-chain with respect to subset inclusion, \ie that no element of this family is a proper subset of another.
    
    Any subset of a maximal context also represents a set of compatible measurements, and we refer to elements of ${\Uc \defeq \setdef{U \subseteq C}{C \in \Mc}}$ as \emph{contexts}.\footnotemark
    
    \footnotetext{While it is more convenient to specify $\Mc$, note that the set of contexts $\Uc$ carries exactly the same information. It forms an abstract simplicial complex whose simplices are the contexts and whose facets are the maximal context. This combinatorial topological structure is emphasised in some presentations \cite{barbosa2014monogamy,barbosa2015contextuality,caru2018towards,karvonen2018categories,abramsky2019comonadic}.}
    
    \item
    $\bfO = \family{\bfO_x}_{x \in X}$ specifies a measurable \emph{space of outcomes} $\bfO_x = \tuple{O_x,\Fc_x}$ for each measurement $x \in X$.
    
    \end{itemize}
\end{definition}

Measurement scenarios can be understood as providing a concise description of the kind of experimental setup that is being considered.
For example, the setup represented in Figure~\ref{fig:scenario} is described by the measurement scenario:
\begin{equation}\label{eq:scenario}
X=\enset{a,a',b,b'} \, , \quad\quad \Mc = \enset{ \, \enset{a,b}, \, \enset{a,b'}, \, \enset{a',b}, \, \enset{a',b'} \, } \, , \quad\quad \bfO_x = \bfR \, ,
\end{equation}
where $\bfR$ is a bounded measurable subspace of $\tuple{\R,\Bc_\R}$.

If some set of measurements $U \subseteq X$ is considered together, there is a joint outcome space given by the product of the respective outcome spaces 
(see Eq.~\eqref{eq:infiniteproduct}),

\begin{equation*}\label{eq:events}
\bfO_U \defeq \prod_{x \in U} \bfO_x = \tuple{O_U, \Fc_U} = \tuple{ \prod_{x \in U} O_x \; , \;  \bigotimes_{x \in U} \Fc_x }  \Mdot
\end{equation*}

The map $\Evs$ that maps $U \subseteq X$ to $\Evs(U) = \bfO_U$
is called the \emph{event sheaf} as concretely it assigns to any 
set of measurements information about the outcome events that could result from jointly performing them.
Note that as well as applying the map to valid contexts $U \in \Uc$ 
we will see that it can also be of interest to consider hypothetical outcome spaces
for sets of measurements that do not necessarily form valid contexts,
in particular $\Evs(X) = \bfO_X$, the joint outcome space for \textit{all} measurements.
Moreover, as we will briefly discuss,
this map satisfies the conditions to be a sheaf ${\fdec{\Evs}{\ps(X)^\op}{\catMeas}}$,
where $\ps(X)$ denotes the powerset of $X$, similarly to its discrete-variable analogue in 
\cite{ab}. 

\subsection{The language of sheaves}
Sheaves are widely used in modern mathematics.
They might roughly be thought of as providing a means of assigning information to the open sets of a topological space in such a way that information can be restricted to smaller open sets and consistent information on a family of open sets can be uniquely `glued' on their union\footnotemark.
In this work we are concerned with discrete topological spaces whose points represent measurements,
and the information that we are interested in assigning has to do with outcome spaces for these measurements and probability measures on these outcome spaces.
Sheaves can be defined concisely in category-theoretic terms as contravariant functors (presheaves) satisfying an additional gluing condition, though in what follows we will also give a more concrete description in terms of restriction maps.
Categorically, the event sheaf is a functor $\fdec{\Evs}{\ps(X)^\op}{\Meas}$ where $\ps(X)$ is viewed as a category in the standard way for partial orders, with morphisms corresponding to subset inclusions.

\footnotetext{For a comprehensive reference on sheaf theory see \eg \cite{maclane1992sheaves}.}

Sheaves come with a notion of \emph{restriction}.
In our example, restriction arises in the following way: whenever $U, V \in \ps(X)$ with $U \subseteq V$ we have an obvious restriction map $\fdec{\rho^V_U}{\Evs(V)}{\Evs(U)}$ which simply projects from the product outcome space for $V$ to that for $U$.
Note that $\rho^U_U$ is the identity map for any $U \in \ps(X)$ and that if $U \subseteq V \subseteq W$ in $\ps(X)$ then $\rho^V_U \circ \rho^W_V = \rho^W_U$.
Already this is enough to show that $\Evs$ is a \stress{presheaf}.
In categorical terms it establishes functoriality.
Our map assigns outcome spaces $\Evs(U) = \bfO_U$ to sets of measurements $U \in \ps(X)$, and in sheaf-theoretic terminology elements of these outcome spaces are called \emph{sections} over $U$.
Sections over $X$ are called \emph{global sections}.
For an inclusion $U \subseteq V$ and a section $\veco \in \Evs(V) = O_V$, it is often more convenient to use the notation $\veco|_U$ to denote $\rho^V_U (\veco) \in \Evs(U) = O_U$, the restriction of $\veco$ to $U$.

Additionally, the unique gluing property holds for $\Evs$.
Suppose that $\Nc \subseteq \ps(X)$ and we have an $\Nc$-indexed family of sections $\family{\veco_U \in O_U}_{U \in \Nc}$ that is compatible in the sense that its elements agree on overlaps, \ie that for all $U, V \in \Nc$, $\veco_U|_{U \cap V} = \veco_V|_{U \cap V}$. Then these sections can always be `glued' together in a unique fashion to obtain a section $\veco_{N}$ over $N \defeq \bigcup \Nc$ such that $\veco_N|_U = \veco_U$ for all $U \in \Nc$.
This makes $\Evs$ a \emph{sheaf}.

We will primarily be concerned with probability measures on outcome spaces.
For this, we recall that the Giry monad $\fdec{\PMeasures}{\catMeas}{\catMeas}$ takes a measurable space and returns the probability measures over that space.
Composing it with the event sheaf yields the map $\PMeasures \circ \Evs$ that takes any context and returns the probability measures on its joint outcome space.
In fact, this is a presheaf $\fdec{\PMeasures \circ \Evs}{\ps(X)^\op}{\catMeas}$,
where restriction on sections is given by marginalisation of probability measures.
Note that marginalisation simply corresponds to the push-forward of a measure along projections to a component of the product space, which are precisely the restriction maps of $\Evs$.
Note, however, that this presheaf does not satisfy the gluing condition and thus it crucially is not a sheaf.

\subsection{Empirical models}

\begin{definition}
An \emph{empirical model} on a measurement scenario $\XMOO$ is a compatible family for the presheaf $\PMeasures\circ\Evs$ on the cover $\Mc$.
Concretely, it is a family $e = \family{e_C}_{C \in \Mc}$, where $e_C$ is a probability measure on the space $\Evs(C)=\bfO_C$ for each maximal context $C \in \Mc$, which satisfies the compatibility condition:
\[e_C|_{C \cap C'} = e_{C'}|_{C \cap C'} \Mdot \]
\end{definition}

Empirical models capture in a precise way the probabilistic \emph{behaviours} that may arise upon performing measurements on physical systems.
The compatibility condition ensures that the empirical behaviour of a given measurement or compatible subset of measurements is independent of which other compatible measurements might be performed along with them.
This is sometimes referred to as the \emph{no-disturbance} condition.
A special case is \emph{no-signalling}, which applies in multi-party or Bell scenarios such as that of Figure~\ref{fig:scenario} and Eq.~\eqref{eq:scenario}.
In that case, contexts consist of measurements that are supposed to occur in  space-like separated locations, and compatibility ensures for instance that the choice of performing $a$ or $a'$ at the first location does not affect the empirical behaviour at the second location, \ie $e_{\enset{a,b}}|_{\enset{b}} = e_{\enset{a',b}}|_{\enset{b}}$.

Note also that while empirical models may arise from the predictions of quantum theory, their definition is theory-independent.
This means that empirical models can just as well describe hypothetical behaviours beyond what can be achieved by quantum mechanics such as the well-studied Popescu--Rohrlich box \cite{popescu1994quantum}.
This can be useful in probing the limits of quantum theory
and in singling out what distinguishes and characterises quantum theory within larger spaces of probabilistic theories,
both well-established lines of research in quantum foundations.

\paragraph{Sheaf-theoretically.}
An empirical model is a compatible family of sections for the presheaf $\PMeasures \circ \Evs$ indexed by the maximal contexts of the measurement scenario.
A natural question that may occur at this point is whether these sections can be glued to form a global section, and this is what we address next.

\subsection{Extendability and contextuality}

\begin{definition}\label{def:nc}
An empirical model $e$ on a scenario $\XMOO$ is \emph{extendable} (or \emph{noncontextual}\footnote{In the language of \cite{griffiths2019quantum}, (non)contextuality here -- and throughout this article -- refers to global (non)contextuality as opposed to Bell (non)contextuality.})
if there is a probability measure $\mu$ on the space $\Evs(X)=\bfO_X$
such that $\mu|_C = e_C$ for every $C \in \Mc$.\footnotemark
\end{definition}

\footnotetext{Notions of partial extendability have also been considered in the discrete setting in \cite{mansfield2014extendability,simmons2017maximally}.}

Recall that $\bfO_X$ is the global outcome space, whose elements correspond to global assignments of outcomes to all the measurements in the given scenario.
Of course, it is not always the case that $X$ is a valid context, and if it were then $\mu = e_X$ would trivially extend the empirical model.
The question of the existence of such a probability measure that recovers the context-wise empirical content of $e$ is particularly significant. When it exists, it amounts to a way of modelling the observed behaviour as arising stochastically from the behaviours of underlying states, identified with the elements of $O_X$, each of which deterministically assigns outcomes to \stress{all} the measurements in $X$ independently of the measurement context that is actually performed.
If an empirical model is not extendable it is said to be \emph{contextual}.
Furthermore, we will say that it is \emph{Bell nonlocal} in the special setting of so-called Bell scenarios, where the compatibility structure of observables is obtained from space-like separation.

\paragraph{Sheaf-theoretically.}
A contextual empirical model is a compatible family of sections for the presheaf $\PMeasures \circ \Evs$ over the contexts of the measurement scenario that cannot be glued into a global section.
Contextuality thus arises as the tension between local consistency and global inconsistency.

    \section{A FAB theorem}\label{sec:fine}
    Quantum theory presents a number of non-intuitive features.
For instance, Einstein, Podolsky and Rosen (EPR) identified early on that if the quantum description of the world is taken as  fundamental  then entanglement poses a problem of ``spooky action at a distance'' \cite{einstein1935can}.
Their conclusion was that quantum theory should be consistent with a deeper or more complete description of the physical world, in which such problems would disappear.
The import of seminal foundational results like the Bell \cite{bell1964einstein} and Bell--Kochen--Specker \cite{bell1966,ks} theorems is that they identify such non-intuitive behaviours and then rule out the possibility of finding \stress{any} underlying model for them that would not suffer from the same issues.
Incidentally, we note that the EPR paradox was originally presented in terms of continuous variables, whereas Bell's theorem addressed a discrete variable analogue of it.

In the previous section, we characterised contextuality of an empirical model by the absence of a global section for that empirical model.
We also saw that global sections capture one kind of underlying model for the behaviour,
namely via deterministic global states that assign predefined outcomes to all measurements. This is precisely the kind of model referred to in the Kochen--Specker theorem \cite{ks}.
Bell's theorem, on the other hand, pertains to a different kind of hidden-variable model, where the salient feature -- Bell locality -- is a kind of factorisability rather than determinism.
Fine \cite{fine1982hidden} showed that in one important measurement scenario (that of the concrete example from Fig.~\ref{fig:scenario}) the existence of one kind of model is equivalent to existence of the other.
Abramsky and Brandenburger \cite{ab} proved in full generality that this existential equivalence holds for any discrete-variable measurement scenario, and that global sections
of $\PMeasures \circ \Evs$
provide a canonical form of hidden-variable model.

In this section, we prove a Fine--Abramsky--Brandenburger theorem in the continuous-variable setting.
It establishes that in this setting there is also an unambiguous, unified description of Bell locality and noncontextuality, which is captured in a canonical way by the notion of extendability.

We will begin by introducing hidden-variable models in a more precise way.
The idea is that there exists some space $\bfLambda$ of hidden variables,
which determine the empirical behaviour.
However, elements of this space may not be directly empirically accessible themselves,
so we allow that we might only have
probabilistic information about them in the form of
a probability measure $p$ on $\bfLambda$.
The empirically observable behaviour should then arise
as an average over the hidden-variable behaviours.

\begin{definition}\label{def:hvmodel}
Let $\XMOO$ be a measurement scenario.
A \emph{hidden-variable} model\footnotemark\ on this scenario consists of the following ingredients:
\footnotetext{The alternative term \emph{ontological model} \cite{spekkens2005contextuality} has become widely used in quantum foundations in recent years. It indicates that the hidden variable, sometimes referred to as the ontic state, is supposed to provide an underlying description of the physical world at perhaps a more fundamental level than the empirical-level description via the quantum state for example.}
\begin{itemize}
    \item
    A measurable space $\bfLambda = \tuple{\Lambda,\Fc_\Lambda}$ of \emph{hidden variables}.
    \item
    A probability measure $p$ on $\bfLambda$.
    \item
    For each maximal context $C \in \Mc$,
    a probability kernel 
    $\fdec{k_C}{\bfLambda}{\Ec(C)}$,\footnotemark\  satisfying the following compatibility condition: for any maximal contexts $C, C' \in \Mc$,
    \begin{equation}\label{eq:compatibility_hiddenvar}
    \Forall{\lambda \in \Lambda} \quad k_C(\lambda,-)|_{C \cap C'} = k_{C'}(\lambda,-)|_{C \cap C'} \Mdot
    \end{equation}
\end{itemize}
\end{definition}

\footnotetext{Recall from Section~\ref{sec:preliminaries} that a probability kernel $\fdec{k_C}{\bfLambda}{\Ec(C)}$ is a function
$\fdec{k_C}{\Lambda \times \Fc_C}{[0,1]}$ which is
a measurable function in the first argument
and
a probability measure in the second argument.%
}

\begin{remark}\label{rem:kerneltohv}
Equivalently, we can regard Eq.~(\ref{eq:compatibility_hiddenvar}) as defining a function $\underline{k}$ from $\Lambda$ to the set of empirical models over $\XMOO$.
The function assigns to each $\lambda \in \Lambda$ the empirical model
$\underline{k}(\lambda) \defeq \family{\underline{k}(\lambda)_C}_{C \in \Mc}$, where the correspondence with the definition above is via $\underline{k}(\lambda)_C = k_C(\lambda,\dummy)$.
This function must be `measurable' in $\bfLambda$
in the sense that  $\fdec{\underline{k}(\dummy)_C(B)}{\Lambda}{[0,1]}$
is a measurable function for all $C \in \Mc$ and $B \in \Fc_C$.
\end{remark}

\begin{definition}
Let $\XMOO$ be a measurement scenario and $\tuple{\bfLambda,p,k}$ be a hidden-variable model.
Then the corresponding empirical model $e$ is given as follows:
for any maximal context $C \in \Mc$ and measurable set of joint outcomes $B \in \Fc_C$,
\[e_C(B) = \intg{\Lambda}{k_C(\dummy,B)}{p}  = \intg{\lambda \in \Lambda}{k_C(\lambda,B)}{p(\lambda)}\Mdot\]
\end{definition}

Note that our definition of hidden-variable model assumes the properties known as $\lambda$-independence \cite{Dickson1998} and parameter-independence \cite{Jarrett1984,Shimony1986}.
The former corresponds to the fact that the probability measure $p$ on the hidden-variable space is independent of the measurement context to be performed,
while the latter corresponds to the compatibility condition \eqref{eq:compatibility_hiddenvar},
which also ensures that the corresponding empirical model satisfies no-signalling \cite{brandenburger2013use}.
We refer the reader to 
\cite{BrandenburgerYanofsky2008} for a detailed discussion of these and other properties of hidden-variable models specifically in the case of multi-party Bell scenarios.

The idea behind the introduction of hidden variables is that they could \textit{explain away} some of the more non-intuitive aspects of the empirical predictions of quantum mechanics,
which would then arise as resulting from an incomplete knowledge of the true state of a system rather than
being a fundamental feature.
There is some precedent for this in physical theories: for instance, statistical mechanics -- a probabilistic theory -- admits a deeper,
albeit usually unwieldily complex, description in terms of classical mechanics,
which is purely deterministic.
Therefore, it is desirable to
impose conditions on hidden-variable models which amount
to requiring that they behave in some sense \stress{classically}
when conditioned on each particular value of the hidden variable $\lambda$.
This motivates the notions of deterministic and of factorisable hidden-variable models.

\begin{definition}
A hidden-variable model $\tuple{\bfLambda,p,k}$ is \emph{deterministic} if the probability kernel
$\fdec{k_C(\lambda,\dummy)}{\Fc_C}{[0,1]}$ is a Dirac measure
for every $\lambda \in \Lambda$
and for every maximal context $C \in \Mc$;
in other words, there is an assignment
$\veco\in O_C$ such that $k_C(\lambda,\dummy) = \delta_{\veco}$.
\end{definition}

In general discussions on hidden-variable models (\eg \cite{BrandenburgerYanofsky2008}), the condition above, requiring that each hidden variable determines a unique joint outcome for each measurement context, is sometimes referred to as weak determinism. This is contraposed to strong determinism,
which demands not only that each hidden variable fix a deterministic outcome to each individual measurement,
but that this outcome be independent of the context in which the measurement is performed.
Note, however, that since our definition of hidden-variable models assumes the compatilibity condition \eqref{eq:compatibility_hiddenvar}, \ie parameter-independence, 
both notions of determinism coincide \cite{brandenburger2013use}.

\begin{definition}
A hidden-variable model $\tuple{\bfLambda,p,k}$ is \emph{factorisable} if
$\fdec{k_C(\lambda,\dummy)}{\Fc_C}{[0,1]}$ factorises as a product measure
for every $\lambda \in \Lambda$
and for every maximal context $C \in \Mc$.
That is, for any family of
measurable sets $\family{B_x \in \Fc_x}_{x \in C}$ with $B_x \neq O_x$ only for finitely many $x \in C$,
\[k_C(\lambda,\prod_{x\in C} B_x) = \prod_{x\in C} k_C|_{\enset{x}}(\lambda,B_x)\]
where
$k_C|_{\enset{x}}(\lambda,\dummy)$ is the marginal of the probability measure
$k_C(\lambda,\dummy)$ on $\bfO_C=\prod_{x \in C}\bfO_x$ to the space $\bfO_{\enset{x}} = \bfO_x$.\footnotemark
\end{definition}

\footnotetext{
Note that, due to the assumption of parameter independence (Eq.~\eqref{eq:compatibility_hiddenvar}),
we can unambiguously write $k_x(\lambda,\dummy)$ for
$k_C|_{\enset{x}}(\lambda,\dummy)$, as this marginal is independent of the context $C$ from which one is restricting.
}

\begin{remark}
In other words, if we consider elements of $\Lambda$ as inaccessible `empirical' models -- \ie if we use the alternative definition of hidden-variable models using the map $\underline{k}$ (see Remark~\ref{rem:kerneltohv}) -- then factorisability is the requirement that each of these be factorisable in the sense that
 \[\underline{k}_C(\lambda)\left(\prod_{x\in C} B_x\right) = \prod_{x\in C} \, \underline{k}_C(\lambda)|_{\enset{x}}(B_x)
 \]
 where $\underline{k}_C|_{\enset{x}}(\lambda)$ is the marginal of the probability measure
 $\underline{k}_C(\lambda)$ on $\bfO_C=\prod_{x \in C}\bfO_x$ to the space $\bfO_x$.
\end{remark}

We now prove the continuous-variable analogue of the theorem proved in the discrete probability setting by Abramsky and Brandenburger \cite[Proposition 3.1 and Theorem 8.1]{ab},
generalising the result of Fine \cite{fine1982hidden} to arbitrary measurement scenarios.

In particular, this result shows that the measurable space $\Evs(X) = \bfO_X$ provides a canonical hidden-variable space.
The proof that \ref{it:ext} $\Rightarrow$ \ref{it:det} in the Theorem below shows how a global probability measure extending an empirical model $e$ can be understood as giving a deterministic hidden-variable model with $\bfLambda = \Evs(X)$. Canonicity is then established together with the proof that \ref{it:fac} $\Rightarrow$ \ref{it:ext}, to the effect that if a given empirical model admits any factorisable hidden-variable model then it admits a deterministic model of the form just mentioned (with $\Evs(X)$ being the hidden-variable space).

\begin{theorem}
Let $e$ be an empirical model on a measurement scenario $\XMOO$. The following are equivalent:
\begin{enumerate}[label=(\arabic*)]
\item\label{it:ext} $e$ is extendable;
\item\label{it:det} $e$ admits a realisation by a deterministic hidden-variable model;
\item\label{it:fac} $e$ admits a realisation by a factorisable hidden-variable model.
\end{enumerate}
\end{theorem}
\begin{proof}
We prove the sequence of implications \ref{it:ext} $\Rightarrow$ \ref{it:det} $\Rightarrow$ \ref{it:fac} $\Rightarrow$ \ref{it:ext}.

\textbf{\ref{it:ext} $\Rightarrow$ \ref{it:det}.}
The idea is that $\Ec(X)=\bfO_X$ provides a canonical deterministic hidden-variable space. Suppose that $e$ is extendable to a global probability measure $\mu$. Let us set
\begin{align*}
\bfLambda &\defeq \bfO_X  \\
p &\defeq \mu \\ 
k_C(\vecg,\dummy) &\defeq \delta_{\vecg|_C}
\end{align*}
for all global outcome assignments $\vecg \in O_X$.
This is by construction a deterministic hidden-variable model,
which we claim gives rise to the empirical model $e$.

Let $C \in \Mc$ and write $\fdec{\rho}{\bfO_X}{\bfO_C}$ for the  measurable projection
which, in the event sheaf, is the restriction map $\fdec{\rho^X_C = \Evs(C \subseteq X)}{\Evs(X)}{\Evs(C)}$.

For any $E \in \Fc_C$, we have
\begin{equation}\label{eq:deltaindicator_aux}
k_C(\vecg,E) = \delta_{\vecg|_C}(E) = \delta_{\rho(\vecg)}(E) = \chi_{_E}(\rho(\vecg))
=
(\chi_{_E} \circ \rho)(\vecg)
\end{equation}
and therefore, as required,
\begin{calculation}
\intg{\bfLambda}{k_C(\dummy,E)}{p}
\just={$\bfLambda = \bfO_X$; $p = \mu$; $k_C(\dummy,E) = \chi_{_E} \circ \rho$ by Eq.~\eqref{eq:deltaindicator_aux}}
\intg{\bfO_X}{\chi_{_E} \circ \rho}{\mu}
\just={change of variables, Eq.~\eqref{eq:changeofvariables}}
\intg{\bfO_C}{\chi_{_E}}{\rho_*\mu}
\just={marginalisation for probability measures}
\intg{\bfO_C}{\chi_{_E}}{\mu|_C}
\just={integral of indicator function, Eq.~\eqref{eq:integralindicator}}
\mu|_C(E)
\just={$\mu$ extends the empirical model $e$}
e_C(E) \Mdot
\end{calculation}%

\textbf{\ref{it:det} $\Rightarrow$ \ref{it:fac}.}
We show that if a hidden-variable model $\tuple{\bfLambda,p,k}$ is deterministic then it is also factorisable.
For this, it is sufficient to notice that a Dirac measure
$\delta_{\veco}$ with $\veco \in O_C$ on a product space $\bfO_C=\prod_{x \in C}\bfO_x$
factorises as the product of Dirac measures
\[\delta_{\veco} =
\prod_{x \in C}\delta_{\veco(x)}
=
\prod_{x \in C}\delta_{\veco|_{\enset{x}}} \Mdot
\]



\textbf{\ref{it:fac} $\Rightarrow$ \ref{it:ext}.}
Suppose that $e$ is realised by a factorisable hidden-variable model $\tuple{\bfLambda,p,k}$.
Write $k_x$ for $k_C|_{\enset{x}}$ as in the definition of factorisability.
Define a measure $\mu$ on $\bfO_X$ as follows:
given a family of measurable sets $\family{E_x \in \Fc_x}_{x\in X}$ with
$E_x = O_x$ for all but finitely many $x \in X$,
the value of $\mu$ on the corresponding cylinder, $\prod_{x\in X}E_x$, is given by
\begin{equation}\label{eq:mudef}
\mu\left(\prod_{x\in X}E_x\right) \defeq
\intg{\Lambda}{\left(\prod_{x\in X} k_x(\dummy,E_x)\right)}{p}
=
\intg{\lambda \in \Lambda}{\left(\prod_{x\in X} k_x(\lambda,E_x)\right)}{p(\lambda)}
\Mcomma
\end{equation}
where the product on the right-hand side is a product of finitely many real numbers in the interval $[0,1]$, since $k_x(\lambda,O_x) = 1$ and so $k_x(\lambda,E_x) \neq 1$ for only finitely many $x \in X$.
Note that the $\sigma$-algebra of $\bfO_X$ is the tensor product $\sigma$-algebra $\Fc_X = \bigotimes_{x \in X}\Fc_x$, which is generated by such cylinders; hence the equation above uniquely determines $\mu$ as a measure on $\bfO_X$.

Now, we show that this is a global section for the empirical probabilities.
Let $C \in \Mc$ and consider 
a `cylinder' set $F = \prod_{x\in C}F_x$ with $F_x \in \Fc_X$ and $F_x \neq O_x$ only for finitely many $x \in C$.
Then

\begin{calculation}
\mu|_{C}(F)
\just={definition of marginalisation}
\mu(F \times O_{X \setminus C})
\just={definition of $F$ and $O_{U}$}
\mu(\prod_{x \in C}F_x \times \prod_{x \in X \setminus C}O_{x})
\just={definition of $\mu$, Eq.~\eqref{eq:mudef}}
\intg{\Lambda}{\left(\prod_{x\in C} k_x(\dummy,F_x)\right)\left(\prod_{x \in X \setminus C} k_x(\dummy,O_x)\right)}{p}
\just={$k_x(\lambda,\dummy)$ is a probability measure so $k_x(\lambda,O_x)=1$}
\intg{\Lambda}{\left(\prod_{x\in C} k_x(\dummy,F_x)\right)}{p}
\just={factorisability of the hidden-variable model}
\intg{\Lambda}{k_C(\dummy,\prod_{x\in C} F_x)}{p}
\just={definition of $F$}
\intg{\Lambda}{k_C(\dummy,F)}{p}
\just={$e$ is the empirical model corresponding to $\tuple{\Lambda,p,k}$}
e_C(F)
\end{calculation}

Since the $\sigma$-algebra $\Fc_C$ of $\bfO_C$ is generated by the cylinder sets of the form above and we have seen that $\mu|_C$ agrees with $e_C$ on these sets, we conclude that $\mu|_C = e_C$ as required.
\end{proof}

    \section{Quantifying contextuality}\label{sec:quantifying}
        Beyond questioning whether a given empirical behaviour is contextual or not,
it is also interesting to ask to what \stress{degree} it is contextual.
In discrete-variable scenarios, a very natural measure of contextuality
is the contextual fraction \cite{ab}.
This measure was shown in \cite{abramsky2017contextual} to have
a number of very desirable properties.
It can be calculated using linear programming,
an approach that subsumes the more traditional approach to quantifying
nonlocality and
contextuality using Bell and noncontextuality inequalities in the sense that we
can understand the (dual) linear program as optimising over \stress{all} such inequalities for the scenario in question and returning
the maximum normalised
violation of \stress{any} Bell or noncontextuality inequality achieved by the given empirical model.
Crucially, the contextual fraction was also shown to \stress{quantifiably} relate to
quantum-over-classical advantages in specific informatic tasks \cite{abramsky2017contextual,mansfield2018quantum,linde}.
Moreover it has been demonstrated to be a monotone with respect to the free
operations of resource theories for contextuality
\cite{abramsky2017contextual,duarte2018resource,abramsky2019comonadic}.

In this section, we consider how to carry those ideas
to the continuous-variable setting. 
The formulation as a linear optimisation problem and the attendant correspondence with Bell inequality violations
requires special care as one needs to use infinite linear programming,
necessitating some extra assumptions on the outcome measurable spaces.

\subsection{The contextual fraction}

Asking whether a given behaviour is noncontextual amounts to
asking whether the empirical model is extendable, or in other words whether
it admits a deterministic hidden-variable model.
However, a more refined question to pose is
\stress{what fraction of the behaviour
admits a deterministic hidden-variable model?}
This quantity is what we call the noncontextual fraction.
Similarly, the fraction of the behaviour that is left over
and that can thus be considered \stress{irreducibly} contextual
is what we call the contextual fraction.

\begin{definition}
Let $e$ be an empirical model on the scenario $\XMOO$.
The \emph{noncontextual fraction} of $e$, written $\NCF(e)$, is defined as
\[\sup\setdef{\mu(O_X)}{\mu \in \Measures(\bfO_X), \, \Forall{C \in \Mc} \mu|_C \leq e_C} \Mdot \]
Note that since $e_C \in \PMeasures(\bfO_C)$ for all $C \in \Mc$ it follows that $\NCF(e) \in [0,1]$. The \emph{contextual fraction} of $e$, written $\CF(e)$, is given by $\CF(e) \defeq 1 - \NCF(e)$.
\end{definition}

\subsection{Monotonicity under free operations including binning}

In the discrete-variable setting, the contextual fraction was shown to be a monotone under a number of natural classical operations that 
transform and combine empirical models and control their use as resources, 
therefore constituting the `free' operations of a resource theory of contextuality \cite{abramsky2017contextual,duarte2018resource,abramsky2019comonadic}.

All of the operations defined for discrete variables in \cite{abramsky2017contextual} -- viz. translations of measurements, transformation of outcomes, probabilistic mixing, product, and choice -- carry almost verbatim to our current setting.
One detail is that one must insist that the coarse-graining of outcomes be achieved by (a family of) measurable functions.
A particular example of practical importance is \emph{binning}, which is widely used in continuous-variable quantum information as
a method of discretising data by partitioning the outcome space $\bfO_x$ for each measurement $x \in X$ into
a finite number of `bins', \ie measurable sets.
Note that a binned empirical model is obtained by pushing forward along a family $\family{t_x}_{x\in X}$ of outcome translations $\fdec{t_x}{\bfO_x}{\bfO'_x}$ where $\bfO'_x$ is finite for all $x \in X$.

For the conditional measurement operation introduced in \cite{abramsky2019comonadic}, which allows for adaptive measurement protocols such as those used in measurement-based quantum computation \cite{raussendorf2001one},
one must similarly insist that the map determining the next measurement to perform based on the observed outcome of a previous measurement 
 be a measurable function. Since we are, for the moment, only considering scenarios where the measurements are treated as constituting a discrete set, this amounts to a partition of the outcome space $\bfO_x$ of the first measurement, $x$, into measurable subsets labelled by measurements compatible with $x$, indicating which will be subsequently performed depending on the outcome observed for $x$.

The inequalities establishing monotonicity from
\cite[Theorem 2]{abramsky2017contextual}
also hold for continuous variables.
There is a caveat for the equality formula for the product of two empirical models: 
\begin{equation*}
\NCF(e_1 \otimes e_2) = \NCF(e_1)\NCF(e_2).
\end{equation*} 
Whereas the inequality establishing monotonicity ($\geq$) stills holds in general, the proof establishing the other direction ($\leq$) makes use of duality of linear programs. Therefore, it only holds under the assumptions we will impose in the remainder of this section.

\begin{proposition}
If $e$ is an empirical model, and $e^\text{bin}$ is any discrete-variable empirical model obtained from $e$ by binning, then contextuality of $e^\text{bin}$ witnesses contextuality of $e$, and quantifiably gives a lower bound $\CF(e^\text{bin}) \leq \CF(e)$.
\end{proposition}

\subsection{Assumptions on the outcome spaces}
\label{subsec:assumption1}

In order to phrase the problem of contextuality as an (infinite) linear programming problem and establish the connection with violations of Bell inequalities, we need to impose some conditions on the measurement scenarios, and in particular on the measurable spaces
of outcomes.

First, from now on we assume that we have a \emph{finite number of measurement labels} \ie that $X$ is finite.

Moreover,
we restrict attention to the case where
the outcome space $\bfO_x$ for each measurement $x \in X$
is the Borel measurable space for a compact Hausdorff space,
\ie that the set $O_x$ is a compact space
and $\Fc_x$ is the $\sigma$-algebra generated by its open sets, written $\Borel(O_x)$.
Note that this includes most situations of interest in practice. 
In particular, it includes the case of measurements with outcomes in a bounded subspace of $\R$ or $\R^n$. This is also experimentally motivated since measurement devices are energetically bounded.
The central missing piece is the case of locally compact spaces, in order to include measurements with outcomes in $\R$ or $\R^n$, which is theoretically relevant ($\R$ would be the canonical outcome space for the quadratures of the electromagnetic field, for instance). We address this issue in the next section and show that it reduces to the compact case.

To summarise we make the following two assumptions here (we will slightly relax the second one later):
\begin{enumerate}[label=(\roman*)]
    \item $X$ is a finite set of measurement labels,
    \item for each $x \in X$, the outcome space $\bfO_x$ is a compact Hausdorff space.
\end{enumerate} 

%
%

To obtain an infinite linear program, we need to work with vector spaces. However, probability measures, or even finite or arbitrary measures, do not form one. We will therefore consider the set 
$\FSMeasures(\bfY)$ of \emph{finite signed measures} (a.k.a.~real measures) on a measurable space ${ \bfY = \tuple{Y,\Fc_Y} }$. These are functions $\fdec{\mu}{\Fc_Y}{\R}$ such that $\mu(\emptyset)=0$ and $\mu$ is $\sigma$-additive. In comparison to the definition of a measure, one drops the nonnegativity requirement, but insists that the values be finite.
The set $\FSMeasures(\bfY)$ forms a real vector space which includes the probability measures $\PMeasures(\bfY)$, and total variation gives a norm on this space.
When $Y$ is a compact Hausdorff space and $\bfY = \tuple{Y,\Borel(Y)}$,
the Riesz--Markov--Kakutani representation theorem \cite{kakutani} says that $\FSMeasures(\bfY)$
is a concrete realisation of the
topological dual space of $C(Y,\R)$,
the space of continuous real-valued functions on $Y$.
The duality is given by
$\tuple{\mu,f} \defeq \intg{\bfY}{f}{\mu}$
for $\mu \in \FSMeasures(\bfY)$ and $f \in C(Y,\R)$.
%

\subsection{Linear programming}

Consider an empirical model $e = \family{e_C}_{C \in \Mc}$ on a scenario $\XMO$ satisfying the assumptions discussed above.
Calculation of its noncontextual fraction can be expressed as the infinite linear programming problem \refprog{LP-CF}.
This is our primal linear program; its dual linear program is given by \refprog{DLP-CF}.
In what follows, we will see how to derive the dual and show that the optimal solutions of both programs coincide.
We also refer the interested reader to Appendix \ref{sec:appendix_std_form}
where the programs are expressed in the standard form of infinite linear programming \cite{barvinok_02}.

\leqnomode 
\begin{flalign*}
    \label{prog:LP-CF}
    \tag*{(P-CF)}
    \hspace{3cm}\left\{
    \begin{aligned}
        & \quad \text{Find } \mu \in \FSMeasures(\bm O_X) \\
        & \quad \text{maximising } \mu(O_X) \\
        & \quad \text{subject to:}  \\
        &  \hspace{1cm} \begin{aligned}
            & \forall C \in \Mc,\; \mu|_C \;\leq\; e_C \\
            & \mu \;\geq\; 0 \Mdot
        \end{aligned}
    \end{aligned}
    \right. &&
\end{flalign*}
\begin{flalign*}
    \label{prog:DLP-CF}
    \tag*{(D-CF)}
    \hspace{3cm}\left\{
    \begin{aligned}
        & \quad \text{Find } \family{f_C}_{C \in \Mc} \in \prod_{C \in \Mc} C(O_C) \\
        & \quad \text{minimising } \sum_{C \in \Mc} \intg{O_C}{f_C}{e_C} \\
        & \quad \text{subject to:}  \\
        & \hspace{1cm} \begin{aligned}
            & \sum_{C \in \Mc} f_C \circ \rho^X_C \;\geq\; \mathbf{1}_{O_X} \\
            & \forall C \in \Mc,\; f_C \;\geq\; \mathbf{0}_{O_C} \Mdot
        \end{aligned}
    \end{aligned}
    \right. &&
\end{flalign*}
\reqnomode

\noindent
We have written $\rho^X_C$ for the projection ${O_X}\lto{O_C}$ as before,
and $\mathbf{1}_D$ (resp. $\mathbf{0}_D$)
for the constant function $D \lto \R$ that assigns the number $1$ (resp. $0$) to all elements of its domain $D$; in the above instance, to all $\vecg \in O_X$ (resp. all $\veco \in O_C$).\footnotemark\ We denote the optimal values of problems \refprog{LP-CF} and \refprog{DLP-CF}, respectively, as
$\val{\text{P-CF}}$ and $\val{\text{D-CF}}$. 
They both equal $\NCF(e)$ due to strong duality (see Proposition~\ref{prop:strong_duality} and Appendix \ref{sec:appendix_proof_zeroduality}).

\footnotetext{Note that $\mathbf{1}_D$ is just a simplified notation for the indicator function on the whole domain; \ie $\fdec{\mathbf{1}_D = \chi_{_{D}}}{D}{\R}$. Similarly, $\mathbf{0}_D$ is the indicator function of the empty set; i.e. $\fdec{\mathbf{0}_D = \chi_{_\emptyset}}{D}{\R}$.}

Analogues of these programs have been studied in the discrete-variable setting \cite{abramsky2017contextual}.
Note, however, that in general these continuous-variable linear programs are over infinite-dimensional spaces and thus not practical to compute directly.
For this reason, in Section~\ref{sec:sdp} we will introduce a hierarchy of finite-dimensional
semi-definite programs that approximate the solution of \refprog{LP-CF} to arbitrary precision.

\subsubsection*{Deriving the dual via the Lagrangian}

We now give an explicit derivation of \refprog{DLP-CF} as the dual of \refprog{LP-CF} via the Lagrangian method.
To simplify notation, we set $E_1 \defeq \FSMeasures(\bfO_X)$ and $F_2 \defeq \prod_{C \in \Mc} C(O_C,\R)$ and their convex cones $K_1$ and $K_2^*$ (see Appendix~\ref{sec:appendix_std_form}).
This matches the standard form notation for infinite linear programming of \cite{barvinok_02}, in which we present our programs in Appendix \ref{sec:appendix_std_form}.
Hence we introduce $\lvert \Mc \rvert$ dual variables,
one continuous map $f_C \in C(O_C,\R)$ for each $C\in \Mc$, to account for the constraints $\mu|_{C} \leq e_C$.
From \refprog{LP-CF}, we then define the Lagrangian $\Lc : K_1 \times K_2^* \longrightarrow \R$ as
\[
    \Lc\left(\mu,(f_C)\right) \defeq \underbrace{\mu(O_X)\vphantom{\sum_{C \in \Mc}}}_{\text{objective}} + \underbrace{\sum_{C \in \Mc} \intg{O_C}{f_C}{(e_C - \mu|_C)}}_{\text{constraints}} \Mdot
\]
The primal program \refprog{LP-CF} corresponds to 
\[
    \sup_{\mu \in K_1} \; \inf_{(f_C) \in K_2^*} \; \Lc(\mu,(f_C)) \Mcomma
\]
as the infimum here imposes the constraints that $\mu|_C \leq e_C$ for all $C \in \Mc$, for otherwise the Lagrangian diverges.
If these constraints are satisfied, then because of the infimum, the second term of the Lagrangian vanishes yielding the objective of the primal problem.
To express the dual, which amounts to permuting the infimum and the supremum, we need to rewrite the Lagrangian:
\begin{align*}
    \Lc(\mu,(f_C)) \quad&=\quad \mu(O_X) + \sum_{C \in \Mc} \intg{O_C}{f_C}{(e_C - \mu|_C)} \\
                     \quad&=\quad \intg{O_X}{\mathbf{1}}{\mu} + \sum_{C \in \Mc} \intg{O_C}{f_C}{e_C}  - \sum_{C \in \Mc} \intg{O_C}{f_C}{\mu|_C} \\
                \quad&=\quad \intg{O_X}{\mathbf{1}}{\mu} + \sum_{C \in \Mc} \intg{O_C}{f_C}{e_C}  - \sum_{C \in \Mc} \intg{O_X}{f_C \circ \rho^X_C}{\mu} \\      
                     \quad&=\quad \intg{O_X}{\mathbf{1}}{\mu} + \sum_{C \in \Mc} \intg{O_C}{f_C}{e_C} - \intg{O_X}{\left( \sum_{C \in \Mc} f_C \circ \rho^X_C \right)}{\mu}  \\ 
                     \quad&=\quad \sum_{C \in \Mc} \intg{O_C}{f_C}{e_C} + \intg{O_X}{\left( \mathbf{1} - \sum_{C \in \Mc} f_C \circ \rho^X_C \right)}{\mu} \Mdot
\end{align*}
The dual program \refprog{DLP-CF} indeed corresponds to
\[
    \inf_{(f_C) \in K_2^*} \; \sup_{\mu \in K_1} \; \Lc(\mu,(f_C)) \Mdot
\]
The supremum imposes that $\sum_{C \in \Mc} f_C \circ \rho^X_C \geq \mathbf{1}$ on $O_X$, since otherwise the Lagrangian diverges. If this constraint is satisfied, then the supremum makes the second term vanish yielding the objective of the dual problem \refprog{DLP-CF}.

\subsubsection*{Zero duality gap}
A key result about the noncontextual fraction, which is essential in establishing the connection to Bell inequality violations, is that \refprog{LP-CF} and \refprog{DLP-CF} are strongly dual, in the sense that no gap exists between their optimal values.
Strong duality always holds in finite linear programming, but it does not hold in general for the infinite case.

\begin{proposition}
\label{prop:zero_duality}
Problems {\upshape \refprog{LP-CF}} and {\upshape \refprog{DLP-CF}} have zero duality gap and their optimal values satisfy
{\upshape
\[
    \val{\text{P-CF}} = \val{\text{D-CF}} = \NCF(e)
\]
}
\label{prop:strong_duality}
\end{proposition}
\begin{proof}
This proof relies on \cite[Theorem 7.2]{barvinok_02}.
The complete proof is provided in Appendix~\ref{sec:appendix_proof_zeroduality}. 
Here, we only provide a brief outline. 
Let $E_1 \defeq \FSMeasures(\bfO_X) \times \prod_{C \in \Mc} \FSMeasures(\bfO_C)$ and $E_2 \defeq \prod_{C \in \Mc} \FSMeasures(\bfO_C)$. Strong duality between \refprog{LP-CF} and \refprog{DLP-CF} amounts to showing that the cone
\[\Kc = \setdef{\left(\, (\mu|_C + \nu_C)_{C\in\Mc}, \mu(O_X)\, \right)}{(\mu,(\nu_C)_{C\in\Mc}) \in E_{1+}} \]
is weakly closed in $E_2 \oplus \R$, where:
\begin{equation*}
E_{1+} \defeq \setdef{ (\mu,(\nu_C)_{C \in \Mc}) \in E_1}{\mu \geq 0 \text{ and } \Forall{C \in \Mc} \nu_C \geq 0}\subset E_1. 
\end{equation*}
We do so by considering a sequence $(\mu^k,(\nu_C^k)_C)_{k \in \N}$ in $E_{1+}$ and showing that the accumulation point
\begin{equation*}
    \lim_{k \rightarrow \infty} \left( (\mu^k|_C+\nu^k)_{C\in\Mc}, \mu^k(O_X) \right)
\end{equation*}
belongs to $\Kc$.
\end{proof}
        
    \section{The case of local compactness}\label{sec:localcompact}
        We now focus on cases where the outcome space might be only locally compact. These include most theoretical situations that are of interest in practice. For instance, $\R$ could be the outcome space for the position and momentum operators. 

For each measurement $x \in X$, $\bm O_x$
is supposed to be the Borel measurable space for a second-countable locally compact Hausdorff space,
\ie that the set $O_x$ is equipped with a second-countable locally compact Hausdorff topology
and $\Fc_x$ is the $\sigma$-algebra generated by its open sets, written $\Borel(O_x)$.
Second countability and Hausdorffness of two spaces $Y$ and $Z$ suffice to show that
the Borel $\sigma$-algebra of the product topology is the tensor product of the Borel $\sigma$-algebras, \ie $\Borel(Y \times Z) = \Borel(Y) \otimes \Borel(Z)$
\cite[Lemma 6.4.2 (Vol.~2)]{Bogachev}. Hence, these assumptions
guarantee that $\bm O_U$ for $U \in \Pc(X)$ is the Borel $\sigma$-algebra of the product topology on $O_U= \prod_{x \in U}O_x$. These product spaces are also second-countable, locally compact, and Hausdorff as all three properties are preserved by finite products.
When $Y$ is a second-countable locally compact Hausdorff space and $\bm Y = \tuple{Y,\Borel(Y)}$,
the Riesz--Markov--Kakutani representation theorem \cite{kakutani} says that $\FSMeasures(\bm Y)$
is a concrete realisation of the
topological dual space of $C_0(Y)$,
the space of continuous real-valued functions on $Y$
that vanish at infinity.\footnote{A function $\fdec{f}{Y}{\R}$ on a locally compact space $Y$ is said to \emph{vanish at infinity} if the set $\setdef{y \in Y}{\|f(x)\|\geq\varepsilon}$ is compact for all $\varepsilon>0$.}
The duality is given by
$\tuple{\mu,f} \defeq \intg{\bm Y}{f}{\mu}$
for $\mu \in \FSMeasures(\bm Y)$ and $f \in C_0(Y)$.\footnotemark Note that when $Y$ is compact (as treated above), $C_0(Y) = C(Y)$ as every closed subspace of a compact space is compact.
\footnotetext{This theorem holds more generally for locally compact Hausdorff spaces if one considers only (finite signed) Radon measures, which are measures that play well with the underlying topology. However, second-countability, together  with local compactness and Hausdorffness, guarantees that every Borel measure is Radon \cite[Theorem 7.8]{Folland}.}

Next, we show that we can approximate the linear program \refprog{LP-CF}\footnotemark\ by a slightly modified linear program defined on the space of finite measures on a measurable compact subspace of $\bm O_X$. 
The idea is to approximate to any desired error the mass of a finite measure on a locally compact set by the mass of the same measure on a compact subset. This naturally comes from the notion of \textit{tightness} of a measure.

\footnotetext{Here we will still use the form of the program \refprog{LP-CF} though throughout this subsection one has to keep in mind that it is defined over finite-signed measures on a \textit{locally compact} space rather than a compact space.}

\begin{definition}[tightness of a measure]
A measure $\mu$ on a metric space $U$ is said to be \textit{tight} if for each $\varepsilon > 0$ there exists a compact set $U_{\varepsilon} \subseteq U$ such that $\mu(U \setminus U_{\varepsilon}) < \varepsilon$.
\end{definition}

\noindent Then we need to argue that every measure we will consider is tight. This is a result of the following theorem.

\begin{theorem}[\cite{parthasarathyprobability1967}]
If $S$ is a complete separable metric space, then every finite measure on $S$ is tight.
\label{th:tightness}
\end{theorem}

\noindent For $x \in X$, $\bm O_x$ is a second-countable locally compact Hausdorff space, thus a Polish space, \ie a separable completely metrisable topological space. For this reason, the above theorem applies. We are now ready to state and prove the main theorem of this subsection.

\begin{theorem}
The linear program \refprog{LP-CF} defined over finite-signed measures on a locally compact space can be approximated to any desired precision $\varepsilon$ by a linear program $(\text{\upshape P-CF}^{\text{\upshape CV},\varepsilon})$ defined over finite signed measures on a compact space.
\label{th:approximate_local}
\end{theorem}

\begin{proof} 
Fix $\varepsilon > 0$. Let $C \in \Mc$ be a given context and $x \in C$ a given measurement label within that context. 
Because $e_C$ is a probability measure on $O_C$, the marginal measure $e_C|_{\{x \}}$ is a finite measure on $O_x$. 
Following Theorem~\ref{th:tightness}, $e_C|_{\{x \}}$ is tight and there exists a compact subset $K_x^{\varepsilon,C} \subseteq O_x$ such that: $e_C|_{\{x \}}(O_x \setminus K_x^{\varepsilon,C}) \leq \varepsilon$. 
Importantly there exist proofs that explicitly construct the approximating sets $K_x^{\varepsilon,C}$ (see \cite{Orbanz2011ProbabilityTI}) based on the separability of the underlying spaces. 
It makes this construction feasible in practice and justifies this approach.

We apply this procedure for every context and for all measurements in a context. 
We now define the compact set
\begin{equation*}
O_x^{\varepsilon} \defeq \bigcup\limits_{C \ni x} K_x^{\varepsilon,C}.    
\end{equation*}
The previous definition is essential to ensure a noncontextual cut-off of the outcome set which ensures the good definition of a compact subset for each measurement label independent of the context. For some subset of measurement labels $U \subseteq X$, we define the compact set $O_U^{\varepsilon} \defeq \prod_{x \in U} O_x^{\varepsilon}$. For every context $C \in \Mc$ and for every measurement label $x \in C$, we now have that $K_x^{\varepsilon,C} \subseteq O_x^{\varepsilon}$ and thus $e_C|_{\{x \}}(O_x \setminus O_x^{\varepsilon}) \leq \varepsilon$. Note that due to the compatibility condition, we can write $e_C|_{\{x\}}$ as $e_{\{x\}}$ for any context. 

Let $\mu$ be any feasible solution of \refprog{LP-CF} defined over finite-signed measures on a locally compact space. 
Due to the constraints of \refprog{LP-CF} we have that $\forall x \in X, \, \mu|_{\{x\}} \leq e_{\{x\}}$.
Then:
\begin{align*}
    \mu(O_X \setminus O^{\varepsilon}_X) & = \mu \left( \prod_{x \in X } O_{x} \setminus \prod_{x \in X } O_{x}^{\varepsilon}  \right) \\
    & = \mu \left( \prod_{x \in X } ( O_{x} \setminus O_{x}^{\varepsilon} ) \right) \\
    & = \prod_{x \in X} \mu|_{\{x\}} (O_x \setminus O_x^\varepsilon) \\
    & \leq \prod_{x \in X} e_{\{x\}} (O_x \setminus O_x^\varepsilon) \\
    & \leq \varepsilon^{\vert X \vert} \Mdot
\end{align*}

We now define the linear program $(\text{P-CF}^{\text{CV},\varepsilon})$ which has the same form as \refprog{LP-CF} though the unknown measures are taken from $\FSMeasures(\bm O_X^{\varepsilon})$ where $\bm O_X^{\varepsilon} = \tuple{O_X^{\varepsilon},\Borel(O_X^{\varepsilon})}$. 
We would like to state that $(\text{P-CF}^{\text{CV},\varepsilon})$ approximates \refprog{LP-CF} up to $\varepsilon$; \ie that their values are $\varepsilon$-close. 
The missing ingredient from the previous chain of inequalities is that given an optimal measure $\mu^*$ satisfying \refprog{LP-CF}, we do not know whether an optimal solution $\mu^*_{\varepsilon}$ of $(\text{P-CF}^{\text{CV},\varepsilon})$ is necessarily the restriction of $\mu^*$ to $O_X^{\varepsilon}$. In fact, it is possible that we do not even have a unique optimal solution. However we only need to prove that they have the same mass on $O_X^\varepsilon$, \ie $\mu_{\varepsilon}^*(O_X^{\varepsilon}) = \mu^*|_{O_X^{\varepsilon}}(O_X^{\varepsilon}) $.
For a contradiction, suppose this does not hold. Then because $\mu^*_{\varepsilon}$ is an optimal value of $(\text{P-CF}^{\text{CV},\varepsilon})$, we must have $\mu_{\varepsilon}^*(O_X^{\varepsilon}) > \mu^*|_{O_X^{\varepsilon}}(O_X^{\varepsilon}) $. From this we construct a new measure $\tilde{\mu}$ on $\bm O_X$ which equals $\mu^*_{\varepsilon}$ on $\bm O_X^{\varepsilon}$ and $\mu^*$ on $\bm O_X \setminus O_X^{\varepsilon}$. It satisfies all constraints and furthermore $\tilde{\mu}(O_X) > \mu^*(O_X)$. This contradicts the fact that $\mu^*$ is an optimal solution of \refprog{LP-CF}. Thus necessarily $\mu_{\varepsilon}^*(O_X^{\varepsilon}) = \mu^*|_{O_X^{\varepsilon}}(O_X^{\varepsilon})$.

The linear program $(\text{P-CF}^{\text{CV},\varepsilon})$ defined on a compact space has indeed a value $\varepsilon$-close to the original program \refprog{LP-CF}.
\end{proof}

In conclusion to this section, we can approximate the problem of finding the noncontextual fraction in measurement scenarios whose outcome spaces are locally compact by the same problem defined on compact subspaces. It thus suffices to restrict the study to the case of compact outcome spaces. 
    
    \section{Continuous generalisation of Bell inequalities}\label{sec:bellinequality}
        The dual program \refprog{DLP-CF} is of particular interest in its own right.
As we now show, it can essentially be understood as
computing a continuous-variable
`Bell inequality' that is optimised to the empirical model.
Making the change of variables $\beta_C \defeq |\Mc|^{-1} \mathbf{1}_{O_C}-f_C$ for each
$C \in \Mc$, the dual program \refprog{DLP-CF} transforms to the following.
\leqnomode
\begin{flalign*}
    \label{prog:B-CF}
    \tag*{(B-CF)}
    \hspace{3cm}\left\{
    \begin{aligned}
        & \quad \text{Find } \family{\beta_C}_{C \in \Mc} \in \prod\limits_{C \in \Mc} C(O_C) \\
        & \quad \text{maximising } \sum_{C \in \Mc} \intg{O_C}{\beta_C}{e_C} \\
        & \quad \text{subject to:}  \\
        & \hspace{1cm} \begin{aligned}
            & \sum_{C \in \Mc} \beta_C \circ \rho^X_C \;\leq\; \mathbf{0}_{O_X} \\
            & \forall C \in \Mc,\; \beta_C \;\leq\; \vert \Mc \vert^{-1} \mathbf{1}_{O_C} \Mdot
        \end{aligned}
    \end{aligned}
    \right. &&
\end{flalign*}
\reqnomode
This program directly computes the contextual fraction $\CF(e)$ instead of the noncontextual fraction. 
It maximises, subject to constraints,
the total value obtained by integrating
these functionals context-wise against the empirical model in question.
The first set of constraints---a generalisation of a system of linear inequalities determining a Bell inequality ---
ensures that, for noncontextual empirical models, the value of the program is
at most $0$, since any such model extends to a measure $\mu$ on $\bm O_X$ such
that $\mu(O_X) = 1$.
The final set of constraints acts as a normalisation condition on the
value of the program, ensuring that it takes values in the interval $[0,1]$ for any empirical model.
Any family of functions $\beta = (\beta_C) \in F_2$ satisfying the constraints will thus result in what can be regarded as a
generalised Bell inequality,
\[
    \sum_{C \in \Mc} \intg{O_C}{\beta_C}{e_C} \leq 0 \, ,
\]
which is satisfied by all noncontextual empirical models.

\begin{definition}
A \emph{form} $\mathbf{\beta}$ on a measurement scenario $\XMOO$ is a family $\mathbf{\beta} = (\beta_C)_{C \in \Mc}$ of functions $\beta_C \in C(O_C)$ for all $C \in \Mc$.
Given an empirical model $e$ on $\XMOO$,
the \emph{value} of $\mathbf{\beta}$ on $e$ is 
\[\langle \beta, e \rangle_{_2} \defeq \sum_{C \in \Mc} \intg{O_C}{\beta_C}{e_C} \Mdot\footnotemark\]
The norm of $\mathbf{\beta}$
is given by
\[\|\beta\| \defeq \sum_{C \in \Mc}\|\beta_C\| = \sum_{C\in\Mc}\sup \setdef{\beta_C(\veco)}{\veco\in O_C}\Mdot\]
\end{definition}

\begin{definition}
An \emph{inequality} $(\mathbf{\beta},R)$ on a measurement scenario $\XMOO$ is a form $\mathbf{\beta}$ together with a bound $R \in \R$.
An empirical model $e$ is said to satisfy the inequality if
the value of $\mathbf{\beta}$ on $e$ is below the bound, 
\ie $\langle \mathbf{\beta}, e \rangle_{_2} \leq R$.
\end{definition}

\begin{definition}
An inequality $(\mathbf{\beta},R)$ is said to be a \emph{generalised Bell inequality} if it is satisfied by all noncontextual empirical models, \ie if for any noncontextual model $d$ on $\XMOO$, it holds that 
$\langle \mathbf{\beta}, d \rangle_{_2} \leq R$.
\end{definition}
\noindent A generalised Bell inequality $(\beta,R)$ establishes a bound $\langle e,\beta \rangle_{2}$ amongst noncontextual models $e$. For more general models, the value of $\beta$ on $e$ is only limited by the algebraic bound $\norm{\beta}$.
In the following, we will only consider inequalities $(\beta,R)$ for which $R < \norm{\beta}$ excluding inequalities trivially satisfied by all empirical models.

\begin{definition}
The \emph{normalised violation} of a generalised Bell inequality $(\mathbf{\beta},R)$ by an empirical model $e$ is
\[\frac{\max\enset{0,\langle \beta, e \rangle_{_2} - R}}{\|\beta\|-R} \Mcomma\]
the amount by which its value $\langle \beta, e \rangle_{_2}$ exceeds the bound $R$ normalised by the maximal `algebraic' violation.
\end{definition}

\footnotetext{The notation $\langle \cdot,\cdot \rangle_{_2}$ is further discussed and explained to be a canonical duality in Appendix~\ref{sec:appendix_std_form}.}

The above definition restricts to the usual notions of Bell inequality and noncontextual inequality in the discrete-variable case
and is particularly close to the presentation in \cite{abramsky2017contextual}.
The following theorem also generalises to continuous variables the main result of \cite{abramsky2017contextual}.

\begin{theorem}
Let e be an empirical model.
\begin{enumerate*}[label=(\roman*)]
\item
The normalised violation by $e$ of any generalised Bell inequality is at most $\CF(e)$;
\item \label{it:ch02_BIth}
if $\CF(e) > 0$ then for every $\varepsilon > 0$ there exists a generalised Bell inequality whose normalised violation by $e$ is at least $CF(e)-\varepsilon$.
\end{enumerate*}
\end{theorem}
\begin{proof}
The proof follows directly from the definitions of the linear programs, and from strong duality, \ie the fact that their optimal values coincide (Proposition \ref{prop:zero_duality} below).
\end{proof}

Item \ref{it:ch02_BIth} is slightly modified compared to the discrete analogue because there is no guarantee that there exists an optimal solution for the dual program \refprog{DLP-CF}. In particular, its optimal value might be achieved by a discontinuous function that can be approximated by continuous ones. Hence the modification of \ref{it:ch02_BIth} with a normalised violation $\varepsilon$-close to $\CF(e)$.

    \section{Approximating the contextual fraction with SDPs}\label{sec:sdp}
        In Section \ref{sec:quantifying}, we presented the problem of computing the noncontextual fraction as an infinite linear program.
Although this is of theoretical importance, it does not allow one to directly perform the actual numerical computation of this quantity.
Here we exploit the link between measures and their sequence of moments to derive a hierarchy of truncated finite-dimensional semidefinite programs which are a relaxation of the original primal problem \refprog{LP-CF}. Dual to this vision, we can equivalently exploit the link between positive polynomials and their sum-of-squares representation to derive a hierarchy of semidefinite programs which are a restriction of the dual problem \refprog{DLP-CF}. 
We further prove that the optimal values of the truncated programs
converge monotonically to the noncontextual fraction.
This makes use of global optimisation techniques developed by Lasserre and Parrilo \cite{lasserre10,parrilo2003semidefinite} and further developed in \cite{lasserre2011new}. 
We introduce them in Appendix~\ref{sec:appendix_Lasserrehierarchy} and we strongly recommend reading this appendix to readers unfamiliar with these notions. 
We will use the same notation throughout this section. 
Another extensive and well-presented reference on the subject is \cite{Laurent2009}.
We start by deriving a hierarchy of SDPs to approximate the contextual fraction and then show that it provides a sequence of optimal values that converge to the noncontextual fraction.

\paragraph{Notation and terminology}
We first fix some notation that is also used in Appendix~\ref{sec:appendix_Lasserrehierarchy}. Let $\Rpolm$ denote the ring of real polynomials in the variables $\bm x \in \R^d$, and let $\Rpolkm \subset \Rpolm$ contain those polynomials of total degree at most $k$.
The latter forms a vector space of dimension $s(k) \defeq \binom{d+k}{k}$,
with a canonical basis consisting of monomials
$\bm x^{\bm \alpha} = x_1^{\alpha_1}\cdots x_d^{\alpha_d}$
indexed by the set $\N^d_k \defeq \setdef{\bm \alpha \in \N^d}{\lvert \bm \alpha \rvert \leq k}$ where $\lvert\bm \alpha\rvert\defeq\sum_{i=1}^d\alpha_i$.
Any $ p \in \Rpolkm$
can be expanded in this basis 
as $p(\bm x) = \sum_{\bm \alpha \in \N^d_k} p_{\alpha}\bm x^{\bm \alpha}$ and we write
$\pv \defeq (p_{\bm \alpha}) \in \R^{s(k)}$ for the resulting vector of coefficients.

\subsection{Hierarchy of semidefinite relaxations for computing \texorpdfstring{$\NCF (e)$}{NCF(e)}}

We fix a measurement scenario $\tuple{X,\Mc,\bm O}$ and an empirical model $e$ on this scenario. 
We will restrict our attention to outcome spaces of the form detailed in Section~\ref{subsec:assumption1}.
Let $d = \vert X \vert \in \N_{>0}$ so that $O_X$ is a Borel subset of $\R^d$.
As a prerequisite, we first need to compute the sequences of moments associated to measures $(e_C)_{C \in \Mc}$ derived from the empirical model. For $C \in \Mc$, let $\bm y^{e,C} = (y^{e,C}_{\bm \alpha})_{\bm \alpha \in \N^d}$ be the sequence of all moments of $e_C$. 
For a given $k \in \N$, which will fix the level of the hierarchy, we only need to compute a finite number $s(k)$ of moments for all contexts. These will be the inputs to the program. 

Below, we derive a hierarchy of SDP relaxations for the primal program \refprog{LP-CF} such that their optimal values converge monotonically to $\val{\text{P-CF}} = \NCF(e)$. We start by discussing the assumptions we have to make on the outcome space. Then we derive the hierarchy based first on the primal program and then on the dual program, and we further show that these formulations are indeed dual. Finally, we prove convergence of the hierarchy.

\subsubsection*{Further assumptions on the outcome space?}
We already made the assumptions mentioned in Section~\ref{subsec:assumption1} for the outcome spaces $\bm O = (\bm O_x)_{x \in X}$, noting that they are not restrictive when considering actual applications. However we would like to meet the assumptions detailed in Assumptions~\ref{ass:algebraic} for the global outcome space $O_X$ so that both Theorems~\ref{th:putinar} and \ref{th:representingmeasure} apply in our setting (see Appendix~\ref{sec:appendix_Lasserrehierarchy}). 

Assumption~\ref{ass:algebraic}~\ref{ass:compact} is already met because we have assumed that for all $x \in X$, $O_x \subset \R$ is compact. Recall that the more general case of $O_x$ locally compact can be reduced to the compact case, as seen in Section~\ref{sec:localcompact}.

Let us discuss Assumption~\ref{ass:algebraic}~\ref{ass:semialgebraic}. We have that $O_X = \prod_{x \in X} O_x$ with $O_x \subset \R$ compact. If $O_x$ is disconnected, we can always complete it into a connected space by attributing measure zero to the added parts for all measures $e_C$ whenever $x \in C$. Then because $O_x$ is compact, it is bounded and it can be described by two constant polynomials: there exists $a_x,b_x \in \R$ such that $O_x = \left[a_x,b_x\right]$. This makes $O_X$ a polytope so in particular, it is semi-algebraic. We write it as 
\[
O_X = \setdef{\bm x \in \R^d}{\forall j = 1,\dots,m, \; g_j(\bm x) \geq 0}
\]
for some polynomials $g_j \in \Rpolm$ of degree 1.

As noted in \cite{lasserre10}, Assumption~\ref{ass:algebraic}~\ref{ass:archimidean} is not very restrictive. For instance, it is satisfied when the set is a polytope. This is the case for $O_X$.

Thus there is no need for further assumptions beyond those already assumed in Section~\ref{subsec:assumption1} in order to apply the results presented in Appendix~\ref{sec:appendix_Lasserrehierarchy}. 

\subsubsection*{Relaxation of the primal program} 

The program \refprog{LP-CF} can be relaxed so that a converging hierarchy of SDPs can be derived. The program \refprog{LP-CF} is essentially a maximisation problem on finite-signed Borel measures with additional constraints such as the fact that these are proper measures (\ie they are nonnegative). We will represent a measure by its moment sequence and use conditions for which this moment sequence has a (unique) representing Borel measure (see Appendix~\ref{subsec:moment}). 
We recall the expression of the primal program \refprog{LP-CF}:
\leqnomode
\begin{flalign*}
    \tag*{(P-CF)}
    \hspace{2.5cm}\left\{
    \begin{aligned}
        & \quad \text{Find } \mu \in \FSMeasures(\bm O_X) \\
        & \quad \text{maximising } \mu(O_X) \\
        & \quad \text{subject to:}  \\
        &  \hspace{1cm} \begin{aligned}
            & \forall C \in \Mc,\; \mu|_C \;\leq\; e_C \\
            & \mu \;\geq\; 0 \Mdot
        \end{aligned}
    \end{aligned}
    \right. &&
\end{flalign*}
\reqnomode
From Appendix~\ref{subsec:moment} which culminates at Theorem~\ref{th:representingmeasure}, it can be relaxed for $k\in \N_{>0}$ as:
\leqnomode
\begin{flalign*}
    \label{prog:SDPk-CF}
    \tag*{(SDP-CF$^k$)}
    \hspace{2.5cm}\left\{
    \begin{aligned}
        & \quad \text{Find } \bm y \in \R^{s(2k)} \\
        & \quad \text{maximising } y_{\bm 0} \\
        & \quad \text{subject to:}  \\
        & \hspace{1cm} \begin{aligned}
            & \forall C \in \Mc,\; M_k(\bm y^{e,C} - \bm y|_C ) \succeq 0, \\
            & \forall j=1,\dots,m,\; M_{k-1}(g_j \bm y) \succeq 0, \\
            & M_k(\bm y) \succeq 0 \Mdot
        \end{aligned}
    \end{aligned}
    \right. &&
\end{flalign*}
\reqnomode
The moment matrices $M_k(\bm y)$ and the localising matrices $M_{k-1}(g_j \bm y)$ are defined in Appendix~\ref{sec:appendix_Lasserrehierarchy}.
We consider localising matrices of order $k-1$ rather than $k$ because all $g_j$'s are of degree exactly 1. In this way, the maximum degree matches with that of the moment matrices. In general we have to deal with localising matrices of order $k- \lceil \frac{\text{deg}(g_j)} 2 \rceil$. If $\mu$ is a representing measure on $\bm O_X$ for $\bm y$ then for all contexts $C \in \Mc$, $\bm y|_C$ can be defined through $\bm y$ by requiring that $\bm y|_C$ has representing measure $\mu|_C$. The two last constraints state necessary conditions on the variable $\bm y$ to be moments of some finite Borel measure supported on $\bm O_X$. The first constraint is a relaxation of the constraint $\mu|_C \leq e_C$ for $C \in \Mc$. As expected, \refprog{SDPk-CF} is a semidefinite relaxation of the problem \refprog{LP-CF} so that $\forall k \in \N_{>0}$, $\NCF(e) = \val{\text{P-CF}} \leq \val{\text{SDP-CF}^{k}} $. Moreover $(\val{\text{SDP-CF}^{k}})_k$ is a monotone nonincreasing sequence because more constraints are added as $k$ increases (so that the relaxations are tighter and tighter).

\subsubsection*{Restriction of the dual program}

The program \refprog{DLP-CF} can be restricted so that we can derive a converging hierarchy of SDPs. It is essentially the minimisation of continuous functions for which we require additional constraints such as the fact that they are nonnegative. We will exploit the link between positive polynomials and sum-of-squares representation that is presented in Appendix~\ref{subsec:sos}. 
We recall the expression of the dual program \refprog{DLP-CF}:
\leqnomode
\begin{flalign*}
    \tag*{(D-CF)}
    \hspace{2.5cm}\left\{
    \begin{aligned}
        & \quad \text{Find } \family{f_C}_{C \in \Mc} \in \prod_{C \in \Mc} C(O_C) \\
        & \quad \text{minimising } \intg{O_C}{f_C}{e_C} \\
        & \quad \text{subject to:}  \\
        & \hspace{1cm} \begin{aligned}
            & \sum_{C \in \Mc} f_C \circ \rho^X_C \;\geq\; \mathbf{1}_{O_X} \\
            & \forall C \in \Mc,\; f_C \;\geq\; \mathbf{0}_{O_C} \Mdot
        \end{aligned}
    \end{aligned}
    \right. &&
\end{flalign*}
\reqnomode
As this point we could derive the dual of program \refprog{SDPk-CF} and show that this is indeed a restriction of the above program. For a more symmetric treatment, we restrict the dual program building on Appendix~\ref{subsec:sos} and Theorem~\ref{th:putinar}. Instead of optimising over positive continuous functions, we restrict them to belong to the quadratic module $Q(g)$ and then further to $Q_k(g)$ for some $k \in \N_{>0}$. This requires that the degrees of SOS polynomials are fixed. 
For $k \in \N_{>0}$, we have
\leqnomode
\begin{flalign*}
    \label{prog:DSDPk-CF}
    \tag*{(DSDP-CF$^{k}$)}
    \hspace{2.5cm}\left\{
    \begin{aligned}
        & \quad \text{Find } (p_C)_{C \in \Mc} \subset \SOSkm \text{ and } \family{\sigma_j}_{j=1,\dots,m} \subset \SOSm_{k-1} \hspace{-2cm}\\
        & \quad \text{maximising } \sum_{C \in \Mc} \intg{O_C}{p_C}{e_C} \\
        & \quad \text{subject to:}  \\
        &  \hspace{1cm} \begin{aligned}
            & \sum_{C \in \Mc} p_C \circ \rho^X_C  - \mathbf{1}_{O_X} = \sum_{j=0}^{m} \sigma_j g_j \Mdot
        \end{aligned}
    \end{aligned}
    \right. &&
\end{flalign*}
\reqnomode
\refprog{DSDPk-CF} is a restriction of \refprog{DLP-CF} so that for all $k \in \N_{> 0}$, we have that $\NCF(e) = \val{\text{D-CF}} \leq \val{\text{SDP-CF}^{k}} $. Furthermore, $(\val{\text{SDP-CF}^{k}})_k$ is a monotone nonincreasing sequence.

Problems \refprog{SDPk-CF} and \refprog{DSDPk-CF} are indeed dual programs (see Proposition~\ref{prop:weakduality} in Appendix~\ref{sec:appendix_SDP_duality}).

\subsection{Convergence of the hierarchy of SDPs}

Finally, we prove that the constructed hierarchy provides a sequence of objective values that converges to the noncontextual fraction $\NCF(e)$.

\begin{theorem}\label{th:ch2_SDPconvergence}\index{Contextual fraction!CV}
The optimal values of the hierarchy of semidefinite programs {\upshape\refprog{SDPk-CF}} (resp. {\upshape\refprog{DSDPk-CF}}) provide monotonically decreasing upper bounds converging to the noncontextual fraction $\NCF(e)$ which is the value of {\upshape\refprog{LP-CF}}. That is 
    {\upshape
    \begin{align*}
        \val{\text{SDP-CF}^\text{k}} \; \downarrow \; \val{\text{P-CF}} = \NCF(e) \quad &\text{ as } k \rightarrow \infty \, ,
    \\
        \val{\text{DSDP-CF}^\text{k}} \; \downarrow \; \val{\text{D-CF}} = \NCF(e) \quad &\text{ as } k \rightarrow \infty \Mdot
    \end{align*}
    }
\end{theorem}

\begin{proof}
Because of the strong duality between the original infinite-dimensional linear programs we have
\[
    \val{\text{P-CF}} = \val{\text{D-CF}} = \NCF(e).
\]
Moreover, for all $k \geq 1$, \refprog{SDPk-CF} is a relaxation of \refprog{LP-CF}:
\[
    \val{\text{SDP-CF}^\text{k}} \geq \val{\text{P-CF}}.
\]
And for all $k \geq 1$, \refprog{DSDPk-CF} is a restriction of \refprog{DLP-CF}:
\[
    \val{\text{DSDP-CF}^\text{k}} \geq \val{\text{D-CF}}.
\]
Also for all $k \geq 1$ we have weak duality between \refprog{SDPk-CF} and \refprog{DSDPk-CF} (see Proposition~\ref{prop:weakduality}):
\[
    \val{\text{DSDP-CF}^\text{k}} \geq \val{\text{SDP-CF}^\text{k}}.
\]
Thus for all $k \geq 1$:
\[
    \val{\text{DSDP-CF}^\text{k}} \geq \val{\text{SDP-CF}^\text{k}} \geq \NCF(e).
\]
We already saw that $(\val{\text{SDP-CF}^\text{k}})_k$ and $(\val{\text{DSDP-CF}^\text{k}})_k$ form monotone nonincreasing sequences. We now show that $(\val{\text{DSDP-CF}^\text{CV,k}})_k$ converges to $\NCF(e)$. This is equivalent to showing that we can approximate any feasible solution\footnotemark of program \refprog{DLP-CF} with a solution of \refprog{DSDPk-CF} for a high enough rank $k$. 
\footnotetext{Note that program \refprog{DLP-CF} might not have an optimal solution as it might only has an optimal solution in the closure of the feasible set. In that case, we can always find a sequence of feasible solutions converging to an optimal solution.}

Fix $\varepsilon > 0$ and any feasible solution $(f_C)_{C \in \Mc} \in \prod_{C \in \Mc} C(O_C)$ of \refprog{DLP-CF}. Then for all $C \in \Mc$, $f_C + \frac{\varepsilon}{\vert \Mc \vert}$ is a positive continuous function on $O_C$. Because $O_C$ is compact (see Subsection~\ref{subsec:assumption1}) by the Stone--Weierstrass theorem, $f_C +  \frac{\varepsilon}{\vert \Mc \vert}$ can be approximated by a positive polynomial. 
Thus there exist positive polynomials $p_C^{\varepsilon} \in \Rpolm$ such that for all contexts $C \in \Mc$  we have (in sup norm) that
\begin{equation}
    \norm{f_C + \frac{\varepsilon}{\vert \Mc \vert} - p_C^{\varepsilon}} \leq
    \frac{\varepsilon}{\vert \Mc \vert} 
    \label{eq:approx1}
\end{equation}
and also
\begin{equation}
    \norm{ \left(f_C + \frac{\varepsilon}{\vert \Mc \vert} - p_C^{\varepsilon} \right) \circ \rho^X_C } < \frac{1}{\vert \Mc \vert} \Min{\bm x \in O_X} \left( \sum_{C \in \Mc} \left(f_C + \frac{\varepsilon}{\vert \Mc \vert} \right) \circ \rho^X_C(\bm x) - \bm x \right) \Mcomma
    \label{eq:approx2}
\end{equation}
where the minimum is strictly positive as $\sum_C (f_C +\frac{\varepsilon}{\vert \Mc \vert}) \circ \rho^X_C > \bm 1_{O_X}  $.

From Eq.~(\ref{eq:approx1}), the objective derived with $(p_C^{\varepsilon})_C$ is $\varepsilon$-close to the original objective:
\begin{align*}
    \left| \sum_{C \in \Mc} \intg{O_C}{f_C}{e_C} - \sum_{C \in \Mc} \intg{O_C}{p_C^{\varepsilon}}{e_C}\right| & \leq \sum_{C \in \Mc}  \intg{O_C}{\left| f_C + \frac{\varepsilon}{\vert \Mc \vert} - p_C^{\varepsilon}\right|}{e_C} \\
    & \leq \varepsilon\Mdot
\end{align*}
Also from Eq.~(\ref{eq:approx2}):
\begin{align*}
    \sum_{C \in \Mc} & p_C^{\varepsilon} \circ \rho^X_C - \bm 1\\
    & > \sum_{C \in \Mc} \left(f_C + \frac{\varepsilon}{\vert \Mc \vert} \right) \circ \rho^X_C - \Min{\bm x \in O_X} \left( \sum_{C \in \Mc} \left(f_C + \frac{\varepsilon}{\vert \Mc \vert} \right) \circ \rho^X_C(\bm x) - \bm x \right) - \bm 1 \\
    & \geq 0,
\end{align*}
so that $\sum_{C \in \Mc} p_C^{\varepsilon} \circ \rho^X_C - \bm 1_{O_X}$ is a positive polynomial on $O_X$.
Next, because $O_X$ is of the form required in Assumption~\ref{ass:algebraic}, by Putinar's Positivellensatz\index{Putinar's Positivellensatz} (see Theorem~\ref{th:putinar}), $\sum_{C \in \Mc} p_C^{\varepsilon} \circ \rho^X_C - \bm 1_{O_X}$ belongs to the quadratic module $Q(g)$. Therefore, for a high enough rank $k \in \N$, it is a feasible solution of \refprog{DSDPk-CF} and thus
\[
    \lvert \NCF(e) - \val{\text{DSDP-CF}^\text{k}} \rvert \leq \varepsilon \Mdot
\]
\end{proof}

    \section*{Outlook}

\stress{Logical} forms of contextuality, which are present
at the level of the \stress{possibilistic} rather than \stress{probabilistic}
information contained in an empirical model, remain to be considered
(e.g.~\cite{fritz2009possibilistic,ab,abramsky2013relational,mansfield2012hardy}).
In the discrete setting, these can be treated by analysing
`possibilistic' empirical models obtained by considering the supports of the discrete-variable probability distributions \cite{ab}, which indicate the
elements of an outcome space that occur with non-zero probability.
In general, the notion of support of a measure is not as straightforward,
and the na\"ive approach is not viable since typically all singletons have measure $0$.
Nevertheless, supports can be defined in the setting of Borel measurable spaces,
for instance, which in any case
are the kind of spaces in which we are practically interested, in Sections \ref{sec:quantifying} and \ref{sec:sdp}.

Approaches to contextuality that characterise obstructions to global sections using cohomology have had some success
\cite{abramsky2012cohomology,abramsky2015contextuality,caru2015detecting,caru2017,raussendorf2016cohomological,roumen2017cohomology,okay2017topological,caru2018towards,okay2018cohomological}
and typically apply to logical forms of contextuality.
An interesting prospect is to explore how the present framework may be employed to these ends, and to see whether the continuous-variable setting can open the door to new techniques that can be applied, or whether qualitatively new forms of contextual behaviour may be uncovered.
A related direction to be developed is to understand how our treatment of contextuality can be further extended to continuous measurement spaces as proposed in \cite{cunha2019measures}.

Another direction to be explored is how our continuous-variable framework for contextuality can be extended to apply to more general notions of contextuality that relate not only to measurement contexts but also more broadly to contexts of preparations and transformations as well \cite{spekkens2005contextuality,mansfield2018quantum}, noting that these also admit quantifiable relationships to quantum advantage \cite{mansfield2018quantum,henaut2018tsirelson}.

Indeed, a major motivation to study contextuality is for its connections to quantum-over-classical advantages in informatic tasks.
An important line of questioning is to ask what further connections can be found in the continuous-variable setting, and whether continuous-variable contextuality might
offer advantages that outstrip those achievable with discrete-variable contextual resources.
Note that it is known that infinite-dimensional quantum systems can offer certain
additional advantages beyond finite-dimensional ones \cite{slofstra2016tsirelson}, though the empirical
model that arises in that example is still a discrete-variable one in our sense.

The present work sets the theoretical basis for computational exploration
of continuous-variable contextuality in quantum-mechanical empirical
models.
This, we hope, can provide new insights and inform other avenues to be developed in future work.
It can also be useful in verifying the non-classicality of empirical models.
Numerical implementation of the programs of Section \ref{sec:sdp} is of particular interest.
The hierarchy of semi-definite programs can be used numerically to witness contextuality in continuous-variable experiments.
Even if the time-complexity of the semi-definite program may increase drastically with its degree, a low-degree program can already provide a first witness of contextual behaviour. 

Since our framework for continuous-variable
contextuality is independent of quantum theory itself, it can equally
be applied to `empirical models' that arise in other, non-physical settings.
The discrete-variable framework of \cite{ab} has led to a number of
surprising connections and cross-fertilisations with other fields \cite{abramsky2015contextual},
including natural language \cite{abramsky2014semantic}, relational databases \cite{abramsky2012databases,barbosa2015contextuality}, logic \cite{abramsky2012logical,abramsky2015contextuality,kishida2016logic}, constraint satisfaction \cite{AbramskyGottlobKolaitis2013robust,abramsky2017quantum} and social systems \cite{dzhafarov2016there}.
It may be hoped that similar connections and applications can be found for the present framework to fields in which continuous-variable data is of
central importance.
For instance probability kernels of the kind we have used are also widely employed in machine learning (\eg \cite{hofmann2008kernel}), inviting intriguing questions about how our framework might be used or what advantages contextuality may confer in that setting.

	\vspace{22pt}
\begin{acknowledgements}

The authors thank Robert Booth, Ulysse Chabaud, Marcelo Terra Cunha, and Antoine Oustry for helpful comments and discussions. We also thank Robert Booth for pointing out that the proof of the FAB theorem could be extended straightforwardly to scenarios with an uncountable set of measurement labels.

This work was largely carried out while RSB was based at the Department of Computer Science, University of Oxford and partly at the School of Informatics, University of Edinburgh and at their current affiliation; while TD was based at the School of Informatics, University of Edinburgh; and while SM was at LIP6, Sorbonne Universit{\'e}. 

Financial support from the following is gratefully acknowledged: the Engineering and Physical Sciences Research Council (EPSRC), EP/N018745/1, `Contextuality as a Resource in Quantum Computation' (RSB); 
EPSRC, EP/R044759/1, `Combining Viewpoints in Quantum Theory (Ext.)' (RSB);
the Portuguese Foundation for Science and Technology (FCT -- Funda{\c{c}\~a}o para a Ci{\^e}ncia e a Tecnologia), CEECINST/00062/2018 (RSB);
and the European Union's Horizon 2020 Research and Innovation Programme under the Marie Sk{\l}odowska-Curie Grant Agreement No.~750523, `Resource Sensitive Quantum Computing' (SM).
\end{acknowledgements}

        \pagebreak
        \bibliographystyle{linksen}      
         \bibliography{cv_bib}   

\providecommand{\href}[2]{#2}\begingroup\raggedright\begin{thebibliography}{10}

\bibitem{bell1966}
J.~S. Bell, ``On the problem of hidden variables in quantum mechanics,''
  \href{http://dx.doi.org/10.1103/RevModPhys.38.447}{{\em Reviews of Modern
  Physics} {\bfseries 38}, 447--452 (1966)}.

\bibitem{ks}
S.~Kochen and E.~P. Specker, ``The problem of hidden variables in quantum
  mechanics,'' {\em Journal of Mathematics and Mechanics} {\bfseries 17},
  59--87 (1967).

\bibitem{raussendorf2013contextuality}
R.~Raussendorf, ``Contextuality in measurement-based quantum computation,''
  \href{http://dx.doi.org/10.1103/PhysRevA.88.022322}{{\em Physical Review A}
  {\bfseries 88}, 022322 (2013)}.

\bibitem{howard2014contextuality}
M.~Howard, J.~Wallman, V.~Veitch, and J.~Emerson, ``Contextuality supplies the
  {`magic'} for quantum computation,''
  \href{http://dx.doi.org/10.1038/nature13460}{{\em Nature} {\bfseries 510},
  351--355 (2014)}.

\bibitem{abramsky2017contextual}
S.~Abramsky, R.~S. Barbosa, and S.~Mansfield, ``Contextual fraction as a
  measure of contextuality,''
  \href{http://dx.doi.org/10.1103/PhysRevLett.119.050504}{{\em Physical Review
  Letters} {\bfseries 119}, 050504 (2017)}.

\bibitem{bermejo2017contextuality}
J.~Bermejo-Vega, N.~Delfosse, D.~E. Browne, C.~Okay, and R.~Raussendorf,
  ``Contextuality as a resource for models of quantum computation with
  qubits,'' \href{http://dx.doi.org/10.1103/PhysRevLett.119.120505}{{\em
  Physical Review Letters} {\bfseries 119}, 120505 (2017)}.

\bibitem{abramsky2017quantum}
S.~Abramsky, R.~S. Barbosa, N.~de~Silva, and O.~Zapata, ``The quantum monad on
  relational structures,''
  \href{http://dx.doi.org/10.4230/LIPIcs.MFCS.2017.35}{in {\em 42nd
  International Symposium on Mathematical Foundations of Computer Science (MFCS
  2017)}, K.~G. Larsen, H.~L. Bodlaender, and J.-F. Raskin, eds.}, vol.~83 of
  {\em Leibniz International Proceedings in Informatics (LIPIcs)},
  pp.~35:1--35:19.
\newblock Schloss Dagstuhl--Leibniz-Zentrum fuer Informatik, 2017.

\bibitem{ab}
S.~Abramsky and A.~Brandenburger, ``The sheaf-theoretic structure of
  non-locality and contextuality,''
  \href{http://dx.doi.org/10.1088/1367-2630/13/11/113036}{{\em New Journal of
  Physics} {\bfseries 13}, 113036 (2011)}.

\bibitem{csw}
A.~Cabello, S.~Severini, and A.~Winter, ``Graph-theoretic approach to quantum
  correlations,'' \href{http://dx.doi.org/10.1103/PhysRevLett.112.040401}{{\em
  Physical Review Letters} {\bfseries 112}, 040401 (2014)}.

\bibitem{afls}
A.~Ac{\'\i}n, T.~Fritz, A.~Leverrier, and A.~B. Sainz, ``A combinatorial
  approach to nonlocality and contextuality,''
  \href{http://dx.doi.org/10.1007/s00220-014-2260-1}{{\em Communications in
  Mathematical Physics} {\bfseries 334}, 533--628 (2015)}.

\bibitem{dzhafarov2015contextuality}
E.~N. Dzhafarov, J.~V. Kujala, and V.~H. Cervantes, ``Contextuality-by-Default:
  A brief overview of ideas, concepts, and terminology,''
  \href{http://dx.doi.org/10.1007/978-3-319-28675-4\_2}{in {\em 9th
  International Conference on Quantum Interaction (QI 2015)}, H.~Atmanspacher,
  T.~Filk, and E.~Pothos, eds.}, vol.~9535 of {\em Lecture Notes in Computer
  Science}, pp.~12--23.
\newblock Springer, 2015.

\bibitem{braunstein2005quantum}
S.~L. Braunstein and P.~Van~Loock, ``Quantum information with continuous
  variables,'' \href{http://dx.doi.org/10.1103/RevModPhys.77.513}{{\em Reviews
  of Modern Physics} {\bfseries 77}, 513 (2005)}.

\bibitem{weedbrook2012gaussian}
C.~Weedbrook, S.~Pirandola, R.~Garc{\'\i}a-Patr{\'o}n, N.~J. Cerf, T.~C. Ralph,
  J.~H. Shapiro, and S.~Lloyd, ``Gaussian quantum information,''
  \href{http://dx.doi.org/10.1103/RevModPhys.84.621}{{\em Reviews of Modern
  Physics} {\bfseries 84}, 621 (2012)}.

\bibitem{yoshikawa2016invited}
J.-i. Yoshikawa, S.~Yokoyama, T.~Kaji, C.~Sornphiphatphong, Y.~Shiozawa,
  K.~Makino, and A.~Furusawa, ``Invited Article: Generation of one-million-mode
  continuous-variable cluster state by unlimited time-domain multiplexing,''
  \href{http://dx.doi.org/10.1063/1.4962732}{{\em APL Photonics} {\bfseries 1},
  060801 (2016)}.

\bibitem{raussendorf2001one}
R.~Raussendorf and H.~J. Briegel, ``A one-way quantum computer,''
  \href{http://dx.doi.org/10.1103/PhysRevLett.86.5188}{{\em Physical Review
  Letters} {\bfseries 86}, 5188 (2001)}.

\bibitem{fine1982hidden}
A.~Fine, ``Hidden variables, joint probability, and the Bell inequalities,''
  \href{http://dx.doi.org/10.1103/PhysRevLett.48.291}{{\em Physical Review
  Letters} {\bfseries 48}, 291 (1982)}.

\bibitem{abramsky2019comonadic}
S.~Abramsky, R.~S. Barbosa, M.~Karvonen, and S.~Mansfield, ``A comonadic view
  of simulation and quantum resources,'' 2019.
\newblock To appear in Proceedings of the 34th Annual ACM/IEEE Symposium on
  Logic in Computer Science (LICS 2019).

\bibitem{cunha2019measures}
M.~T. Cunha, ``On measures and measurements: a fibre bundle approach to
  contextuality.'' Preprint
  \href{https://arxiv.org/abs/1903.08819}{arXiv:1903.08819 [math.PR]}, 2019.

\bibitem{plastino2010state}
{\'A}.~R. Plastino and A.~Cabello, ``State-independent quantum contextuality
  for continuous variables,''
  \href{http://dx.doi.org/10.1103/PhysRevA.82.022114}{{\em Physical Review A}
  {\bfseries 82}, 022114 (2010)}.

\bibitem{he2010bell}
Q.-Y. He, E.~G. Cavalcanti, M.~D. Reid, and P.~D. Drummond, ``Bell inequalities
  for continuous-variable measurements,''
  \href{http://dx.doi.org/10.1103/PhysRevA.81.062106}{{\em Physical Review A}
  {\bfseries 81}, 062106 (2010)}.

\bibitem{mckeown2011testing}
G.~McKeown, M.~G. Paris, and M.~Paternostro, ``Testing quantum contextuality of
  continuous-variable states,''
  \href{http://dx.doi.org/10.1103/PhysRevA.83.062105}{{\em Physical Review A}
  {\bfseries 83}, 062105 (2011)}.

\bibitem{su2012quantum}
H.-Y. Su, J.-L. Chen, C.~Wu, S.~Yu, and C.~Oh, ``Quantum contextuality in a
  one-dimensional quantum harmonic oscillator,''
  \href{http://dx.doi.org/10.1103/PhysRevA.85.052126}{{\em Physical Review A}
  {\bfseries 85}, 052126 (2012)}.

\bibitem{asadian2015contextuality}
A.~Asadian, C.~Budroni, F.~E. Steinhoff, P.~Rabl, and O.~G{\"u}hne,
  ``Contextuality in phase space,''
  \href{http://dx.doi.org/10.1103/PhysRevLett.114.250403}{{\em Physical Review
  Letters} {\bfseries 114}, 250403 (2015)}.

\bibitem{laversanne2017general}
A.~Laversanne-Finot, A.~Ketterer, M.~R. Barros, S.~P. Walborn, T.~Coudreau,
  A.~Keller, and P.~Milman, ``General conditions for maximal violation of
  non-contextuality in discrete and continuous variables,''
  \href{http://dx.doi.org/10.1088/1751-8121/aa6016}{{\em Journal of Physics A:
  Mathematical and Theoretical} {\bfseries 50}, 155304 (2017)}.

\bibitem{ketterer2018continuous}
A.~Ketterer, A.~Laversanne-Finot, and L.~Aolita, ``Continuous-variable
  supraquantum nonlocality,''
  \href{http://dx.doi.org/10.1103/PhysRevA.97.012133}{{\em Physical Review A}
  {\bfseries 97}, 012133 (2018)}.

\bibitem{popescu1994quantum}
S.~Popescu and D.~Rohrlich, ``Quantum nonlocality as an axiom,''
  \href{http://dx.doi.org/10.1007/BF02058098}{{\em Foundations of Physics}
  {\bfseries 24}, 379--385 (1994)}.

\bibitem{mansfield2012hardy}
S.~Mansfield and T.~Fritz, ``Hardy’s non-locality paradox and possibilistic
  conditions for non-locality,''
  \href{http://dx.doi.org/10.1007/s10701-012-9640-1}{{\em Foundations of
  Physics} {\bfseries 42}, 709--719 (2012)}.

\bibitem{mansfield2017consequences}
S.~Mansfield, ``Consequences and applications of the completeness of Hardy's
  nonlocality,'' \href{http://dx.doi.org/10.1103/PhysRevA.95.022122}{{\em
  Physical Review A} {\bfseries 95}, 022122 (2017)}.

\bibitem{Bill79}
P.~Billingsley, ``Probability and Measure,''.
\newblock Wiley Series in Probability and Mathematical Statistics. Wiley, 1979.

\bibitem{panangaden2009labelled}
P.~Panangaden, ``Labelled {M}arkov Processes,''.
\newblock \href{http://dx.doi.org/10.1142/p595}{Imperial College Press}, 2009.

\bibitem{vestrup2003measures}
E.~Vestrup, ``The theory of measures and integration,''.
\newblock Wiley Series in Probability and Statistics. John Wiley \& Sons, 2003.

\bibitem{giry1982categorical}
M.~Giry, ``A categorical approach to probability theory,''
  \href{http://dx.doi.org/10.1007/BFb0092872}{in {\em Categorical aspects of
  topology and analysis}, B.~Banaschewski, ed.}, vol.~915 of {\em Lecture Notes
  in Mathematics}, pp.~68--85.
\newblock Springer, 1982.

\bibitem{barbosa2014monogamy}
R.~S. Barbosa, ``On monogamy of non-locality and macroscopic averages: examples
  and preliminary results,'' \href{http://dx.doi.org/10.4204/EPTCS.172.4}{in
  {\em 11th Workshop on Quantum Physics and Logic (QPL 2014)}, B.~Coecke,
  I.~Hasuo, and P.~Panangaden, eds.}, vol.~172 of {\em Electronic Proceedings
  in Theoretical Computer Science}, pp.~36--55.
\newblock Open Publishing Association, 2014.

\bibitem{barbosa2015contextuality}
R.~S. Barbosa, {\em Contextuality in quantum mechanics and beyond}.
\newblock PhD thesis, University of Oxford, 2015.

\bibitem{caru2018towards}
G.~Car{\`u}, ``Towards a complete cohomology invariant for non-locality and
  contextuality.'' Preprint
  \href{https://arxiv.org/abs/1807.04203}{arXiv:1807.04203 [quant-ph]}, 2018.

\bibitem{karvonen2018categories}
M.~Karvonen, ``Categories of Empirical Models,''
  \href{http://dx.doi.org/10.4204/EPTCS.287.14}{in {\em 15th International
  Conference on Quantum Physics and Logic (QPL 2018)}, P.~Selinger and
  G.~Chiribella, eds.}, vol.~287 of {\em Electronic Proceedings in Theoretical
  Computer Science}, pp.~239--252.
\newblock Open Publishing Association, 2019.

\bibitem{maclane1992sheaves}
S.~MacLane and I.~Moerdijk, ``Sheaves in geometry and logic: A first
  introduction to topos theory,''.
\newblock Universitext.
  \href{http://dx.doi.org/10.1007/978-1-4612-0927-0}{Springer}, 1992.

\bibitem{griffiths2019quantum}
R.~B. Griffiths, ``Quantum measurements and contextuality,''
  \href{http://dx.doi.org/10.1098/rsta.2019.0033}{{\em Philosophical
  Transactions of the Royal Society A} {\bfseries 377}, 20190033 (2019)}.

\bibitem{mansfield2014extendability}
S.~Mansfield and R.~S. Barbosa, ``Extendability in the sheaf-theoretic
  approach: Construction of {B}ell models from {K}ochen--{S}pecker models.''
  Preprint \href{https://arxiv.org/abs/1402.4827}{arXiv:1402.4827 [quant-ph]},
  2014.

\bibitem{simmons2017maximally}
A.~W. Simmons, ``How (maximally) contextual is quantum mechanics?.'' Preprint
  \href{https://arxiv.org/abs/1712.03766}{arXiv:1712.03766 [quant-ph]}, 2017.

\bibitem{einstein1935can}
A.~Einstein, B.~Podolsky, and N.~Rosen, ``Can quantum-mechanical description of
  physical reality be considered complete?,''
  \href{http://dx.doi.org/10.1103/PhysRev.47.777}{{\em Physical Review}
  {\bfseries 47}, 777 (1935)}.

\bibitem{bell1964einstein}
J.~S. Bell, ``On the {E}instein {P}odolsky {R}osen paradox,''
  \href{http://dx.doi.org/10.1103/PhysicsPhysiqueFizika.1.195}{{\em Physics
  Physique Fizika} {\bfseries 1}, 195 (1964)}.

\bibitem{spekkens2005contextuality}
R.~W. Spekkens, ``Contextuality for preparations, transformations, and unsharp
  measurements,'' \href{http://dx.doi.org/10.1103/PhysRevA.71.052108}{{\em
  Physical Review A} {\bfseries 71}, 052108 (2005)}.

\bibitem{Dickson1998}
W.~M. Dickson, ``Quantum chance and non-locality: Probability and non-locality
  in the interpretations of quantum mechanics,''.
\newblock \href{http://dx.doi.org/10.1017/CBO9780511524738}{Cambridge
  University Press}, 1998.

\bibitem{Jarrett1984}
J.~P. Jarrett, ``On the physical significance of the locality conditions in the
  {B}ell arguments,'' \href{http://dx.doi.org/10.2307/2214878}{{\em
  {\textnormal{Special issue on the Foundations of Quantum Mechanics}},
  No{\^u}s} {\bfseries 18}, 569--589 (1984)}.

\bibitem{Shimony1986}
A.~Shimony, ``Events and processes in the quantum world,''
  \href{http://dx.doi.org/10.1017/CBO9781139172196.011}{in {\em Quantum
  concepts in space and time}, R.~Penrose and C.~J. Isham, eds.}, pp.~182--203.
\newblock Oxford University Press, 1986.

\bibitem{brandenburger2013use}
A.~Brandenburger and H.~J. Keisler, ``Use of a canonical hidden-variable space
  in quantum mechanics,''
  \href{http://dx.doi.org/10.1007/978-3-642-38164-5\_1}{in {\em Computation,
  Logic, Games, and Quantum Foundations. The Many Facets of Samson Abramsky:
  Essays dedicated to Samson Abramsky on the occasion of his 60th birthday},
  B.~Coecke, L.~Ong, and P.~Panangaden, eds.}, pp.~1--6.
\newblock Springer, 2013.

\bibitem{BrandenburgerYanofsky2008}
A.~Brandenburger and N.~Yanofsky, ``A classification of hidden-variable
  properties,'' \href{http://dx.doi.org/10.1088/1751-8113/41/42/425302}{{\em
  Journal of Physics A: Mathematical and Theoretical} {\bfseries 41}, 425302
  (2008)}.

\bibitem{mansfield2018quantum}
S.~Mansfield and E.~Kashefi, ``Quantum advantage from sequential-transformation
  contextuality,'' \href{http://dx.doi.org/10.1103/PhysRevLett.121.230401}{{\em
  Physical Review Letters} {\bfseries 121}, 230401 (2018)}.

\bibitem{linde}
L.~Wester, {\em Classical and quantum structures of computation}.
\newblock PhD thesis, University of Oxford, 2018.

\bibitem{duarte2018resource}
C.~Duarte and B.~Amaral, ``Resource theory of contextuality for arbitrary
  prepare-and-measure experiments,''
  \href{http://dx.doi.org/10.1063/1.5018582}{{\em Journal of Mathematical
  Physics} {\bfseries 59}, 062202 (2018)}.

\bibitem{kakutani}
S.~Kakutani, ``Concrete representation of abstract {$(M)$}-spaces (A
  characterization of the space of continuous functions),''
  \href{http://dx.doi.org/10.2307/1968778}{{\em Annals of Mathematics}
  {\bfseries 42}, 994--1024 (1941)}.

\bibitem{barvinok_02}
A.~Barvinok, ``A course in convexity,'', vol.~54 of {\em Graduate Studies in
  Mathematics}.
\newblock \href{http://dx.doi.org/10.1090/gsm/054}{American Mathematical
  Society}, 2002.

\bibitem{Bogachev}
V.~I. Bogachev, ``Measure theory,''.
\newblock \href{http://dx.doi.org/10.1007/978-3-540-34514-5}{Springer}, 2007.

\bibitem{Folland}
G.~B. Folland, ``Real Analysis: Modern techniques and their applications,''.
\newblock Pure and Applied Mathematics. John Wiley \& Sons, 1984.
\newblock 2nd Edition (1999).

\bibitem{parthasarathyprobability1967}
K.~Parthasarathy, ``Probability {Measures} on {Metric} {Spaces},''.
\newblock \href{http://dx.doi.org/10.1016/C2013-0-08107-8}{Elsevier}, 1967.

\bibitem{Orbanz2011ProbabilityTI}
P.~Orbanz, ``Probability Theory II,'' 2011.

\bibitem{lasserre10}
J.-B. Lasserre, ``Moments, positive polynomials and their applications,'',
  vol.~1 of {\em Series on Optmization and Its Applications}.
\newblock \href{http://dx.doi.org/10.1142/p665}{Imperial College Press}, 2009.

\bibitem{parrilo2003semidefinite}
P.~A. Parrilo, ``Semidefinite programming relaxations for semialgebraic
  problems,'' \href{http://dx.doi.org/10.1007/s10107-003-0387-5}{{\em
  Mathematical programming} {\bfseries 96}, 293--320 (2003)}.

\bibitem{lasserre2011new}
J.~B. Lasserre, ``A new look at nonnegativity on closed sets and polynomial
  optimization,'' \href{http://dx.doi.org/10.1137/100806990}{{\em SIAM Journal
  on Optimization} {\bfseries 21}, 864--885 (2011)}.

\bibitem{Laurent2009}
M.~Laurent, {\em Sums of Squares, Moment Matrices and Optimization Over
  Polynomials},
  \href{http://dx.doi.org/10.1007/978-0-387-09686-5_7}{pp.~157--270}.
\newblock Springer New York, New York, NY, 2009.

\bibitem{fritz2009possibilistic}
T.~Fritz, ``Possibilistic Physics,'' 2009.
\newblock FQXI Essay Contest 2009.

\bibitem{abramsky2013relational}
S.~Abramsky, ``Relational hidden variables and non-locality,''
  \href{http://dx.doi.org/10.1007/s11225-013-9477-4}{{\em Studia Logica}
  {\bfseries 101}, 411--452 (2013)}.

\bibitem{abramsky2012cohomology}
S.~Abramsky, S.~Mansfield, and R.~S. Barbosa, ``The cohomology of non-locality
  and contextuality,'' \href{http://dx.doi.org/10.4204/EPTCS.95.1}{in {\em
  Proceedings of 8th International Workshop on Quantum Physics and Logic (QPL
  2011)}, B.~Jacobs, P.~Selinger, and B.~Spitters, eds.}, vol.~95 of {\em
  Electronic Proceedings in Theoretical Computer Science}, pp.~1--14.
\newblock Open Publishing Association, 2012.

\bibitem{abramsky2015contextuality}
S.~Abramsky, R.~S. Barbosa, K.~Kishida, R.~Lal, and S.~Mansfield,
  ``Contextuality, cohomology and paradox,''
  \href{http://dx.doi.org/10.4230/LIPIcs.CSL.2015.211}{in {\em 24th EACSL
  Annual Conference on Computer Science Logic (CSL 2015)}, S.~Kreutzer, ed.},
  vol.~41 of {\em Leibniz International Proceedings in Informatics (LIPIcs)},
  pp.~211--228.
\newblock Schloss Dagstuhl--Leibniz-Zentrum fuer Informatik, 2015.

\bibitem{caru2015detecting}
G.~Car{\`u}, ``Detecting contextuality: Sheaf cohomology and All vs Nothing
  arguments,'' Master's thesis, University of Oxford, 2015.
\newblock Available at
  \href{http://www.cs.ox.ac.uk/files/7608/Dissertation.pdf}{\url{http://www.cs.ox.ac.uk/files/7608/Dissertation.pdf}}.

\bibitem{caru2017}
G.~Car{\`u}, ``On the cohomology of contextuality,''
  \href{http://dx.doi.org/10.4204/EPTCS.236.2}{in {\em 13th International
  Conference on Quantum Physics and Logic (QPL 2016)}, R.~Duncan and C.~Heunen,
  eds.}, vol.~236 of {\em Electronic Proceedings in Theoretical Computer
  Science}, pp.~21--39.
\newblock Open Publishing Association, 2017.

\bibitem{raussendorf2016cohomological}
R.~Raussendorf, ``Cohomological framework for contextual quantum
  computations.'' Preprint
  href{https://arxiv.org/abs/1602.04155}{arXiv:1602.04155 [quant-ph]}, 2016.

\bibitem{roumen2017cohomology}
F.~Roumen, ``Cohomology of effect algebras,''
  \href{http://dx.doi.org/10.4204/EPTCS.236.12}{in {\em 13th International
  Conference on Quantum Physics and Logic (QPL 2016)}, R.~Duncan and C.~Heunen,
  eds.}, vol.~236 of {\em Electronic Proceedings in Theoretical Computer
  Science}, pp.~174--201.
\newblock Open Publishing Association, 2017.

\bibitem{okay2017topological}
C.~Okay, S.~Roberts, S.~D. Bartlett, and R.~Raussendorf, ``Topological proofs
  of contextuality in quantum mechanics,''
  \href{http://dx.doi.org/10.26421/QIC17.13-14}{{\em Quantum Information and
  Computation} {\bfseries 17}, 1135--1166 (2017)}.

\bibitem{okay2018cohomological}
C.~Okay, E.~Tyhurst, and R.~Raussendorf, ``The cohomological and the
  resource-theoretic perspective on quantum contextuality: common ground
  through the contextual fraction.'' Preprint
  \href{https://arxiv.org/abs/1806.04657}{arXiv:1806.04657 [quant-ph]}, 2018.

\bibitem{henaut2018tsirelson}
L.~Henaut, L.~Catani, D.~E. Browne, S.~Mansfield, and A.~Pappa, ``Tsirelson's
  bound and {L}andauer's principle in a single-system game,''
  \href{http://dx.doi.org/10.1103/PhysRevA.98.060302}{{\em Physical Review A}
  {\bfseries 98}, 060302 (2018)}.

\bibitem{slofstra2016tsirelson}
W.~Slofstra, ``Tsirelson's problem and an embedding theorem for groups arising
  from non-local games.'' Preprint
  \href{https://arxiv.org/abs/1606.03140}{arXiv:1606.03140 [quant-ph]}, 2016.

\bibitem{abramsky2015contextual}
S.~Abramsky, ``Contextual semantics: From quantum mechanics to logic,
  databases, constraints, and complexity,''
  \href{http://dx.doi.org/10.1142/9789814730617\_0002}{in {\em Contextuality
  from Quantum Physics to Psychology}, E.~Dzhafarov, S.~Jordan, R.~Zhang, and
  V.~Cervantes, eds.}, vol.~6 of {\em Advanced Series on Mathematical
  Psychology}, pp.~23--50.
\newblock World Scientific, 2015.

\bibitem{abramsky2014semantic}
S.~Abramsky and M.~Sadrzadeh, ``Semantic unification: A sheaf-theoretic
  approach to natural language,''
  \href{http://dx.doi.org/10.1007/978-3-642-54789-8\_1}{in {\em Categories and
  Types in Logic, Language, and Physics: Essays dedicated to {J}im {L}ambek on
  the occasion of his 90th birthday}, C.~Casadio, B.~Coecke, M.~Moortgat, and
  P.~Scott, eds.}, pp.~1--13.
\newblock Springer, 2014.

\bibitem{abramsky2012databases}
S.~Abramsky, ``Relational databases and {B}ell's theorem,''
  \href{http://dx.doi.org/10.1007/978-3-642-41660-6\_2}{in {\em In search of
  elegance in the theory and practice of computation: Essays dedicated to
  {P}eter {B}uneman}, V.~Tannen, L.~Wong, L.~Libkin, W.~Fan, W.-C. Tan, and
  M.~Fourman, eds.}, vol.~8000 of {\em Lecture Notes in Computer Science},
  pp.~13--35.
\newblock Springer, 2013.

\bibitem{abramsky2012logical}
S.~Abramsky and L.~Hardy, ``Logical {B}ell inequalities,''
  \href{http://dx.doi.org/10.1103/PhysRevA.85.062114}{{\em Physical Review A}
  {\bfseries 85}, 062114 (2012)}.

\bibitem{kishida2016logic}
K.~Kishida, ``Logic of local inference for contextuality in quantum physics and
  beyond,'' \href{http://dx.doi.org/10.4230/LIPIcs.ICALP.2016.113}{in {\em 43rd
  International Colloquium on Automata, Languages, and Programming (ICALP
  2016)}, I.~Chatzigiannakis, M.~Mitzenmacher, Y.~Rabani, and D.~Sangiorgi,
  eds.}, vol.~55 of {\em Leibniz International Proceedings in Informatics
  (LIPIcs)}, pp.~113:1--113:14.
\newblock Schloss Dagstuhl--Leibniz-Zentrum fuer Informatik, 2016.

\bibitem{AbramskyGottlobKolaitis2013robust}
S.~Abramsky, G.~Gottlob, and P.~G. Kolaitis, ``Robust constraint satisfaction
  and local hidden variables in quantum mechanics,'' in {\em 23rd International
  Joint Conference on Artificial Intelligence (IJCAI 2013)}, F.~Rossi, ed.,
  pp.~440--446, AAAI Press.
\newblock 2013.

\bibitem{dzhafarov2016there}
E.~N. Dzhafarov, R.~Zhang, and J.~Kujala, ``Is there contextuality in
  behavioural and social systems?,''
  \href{http://dx.doi.org/10.1098/rsta.2015.0099}{{\em {\textnormal{Theme issue
  on Quantum probability and the mathematical modelling of decision making}},
  Philosophical Transactions of the Royal Society A: Mathematical, Physical and
  Engineering Sciences} {\bfseries 374}, 20150099 (2016)}.

\bibitem{hofmann2008kernel}
T.~Hofmann, B.~Sch{\"o}lkopf, and A.~J. Smola, ``Kernel methods in machine
  learning,'' \href{http://dx.doi.org/10.1214/009053607000000677}{{\em The
  Annals of Statistics} {\bfseries 36}, 1171--1220 (2008)}.

\bibitem{luenberger97}
D.~G. Luenberger, ``Optimization by vector space methods,''.
\newblock John Wiley \& Sons, 1997.

\bibitem{lasserre2001global}
J.-B. Lasserre, ``Global optimization with polynomials and the problem of
  moments,'' \href{http://dx.doi.org/10.1137/S1052623400366802}{{\em SIAM
  Journal on optimization} {\bfseries 11}, 796--817 (2001)}.

\bibitem{hilbert1888darstellung}
D.~Hilbert, ``{\"U}ber die darstellung definiter formen als summe von
  formenquadraten,'' \href{http://dx.doi.org/10.1007/BF01443605}{{\em
  Mathematische Annalen} {\bfseries 32}, 342--350 (1888)}.

\bibitem{Putinar93}
M.~Putinar, ``Positive Polynomials on Compact Semi-algebraic Sets,'' {\em
  Indiana University Mathematics Journal} {\bfseries 42}, 969--984 (1993).

\bibitem{haviland1936momentum}
E.~Haviland, ``On the momentum problem for distribution functions in more than
  one dimension. II,'' \href{http://dx.doi.org/10.2307/2371063}{{\em American
  Journal of Mathematics} {\bfseries 58}, 164--168 (1936)}.

\end{thebibliography}\endgroup

	\newpage
	\section*{Appendices}\label{sec:appendix}
		\appendix

\section{Linear programs in standard form}
\label{sec:appendix_std_form}

This appendix may be of particular interest to readers familiar with global optimisation. We express the problems \refprog{LP-CF} and \refprog{DLP-CF} in the standard form of infinite linear programming \cite[IV--(7.1)]{barvinok_02}. We recall them below:

\leqnomode 
\begin{flalign*}
    \label{prog:LP-CF-re}
    \tag*{(P-CF)}
    \hspace{3cm}\left\{
    \begin{aligned}
        & \quad \text{Find } \mu \in \FSMeasures(\bm O_X) \\
        & \quad \text{maximising } \mu(O_X) \\
        & \quad \text{subject to:}  \\
        &  \hspace{1cm} \begin{aligned}
            & \forall C \in \Mc,\; \mu|_C \;\leq\; e_C \\
            & \mu \;\geq\; 0 \Mdot
        \end{aligned}
    \end{aligned}
    \right. &&
\end{flalign*}
\begin{flalign*}
    \label{prog:DLP-CF-re}
    \tag*{(D-CF)}
    \hspace{3cm}\left\{
    \begin{aligned}
        & \quad \text{Find } \family{f_C}_{C \in \Mc} \in \prod_{C \in \Mc} C(O_C) \\
        & \quad \text{minimising } \intg{O_C}{f_C}{e_C} \\
        & \quad \text{subject to:}  \\
        & \hspace{1cm} \begin{aligned}
            & \sum_{C \in \Mc} f_C \circ \rho^X_C \;\geq\; \mathbf{1}_{O_X} \\
            & \forall C \in \Mc,\; f_C \;\geq\; \mathbf{0}_{O_C} \Mdot
        \end{aligned}
    \end{aligned}
    \right. &&
\end{flalign*}
\reqnomode

Problems \refprog{LP-CF-re} and \refprog{DLP-CF-re} are indeed infinite linear programs as both the objective and the constraints are linear 
with respect to the unknown measure $\mu \in \FSMeasures(\bfO_X)$. To write \refprog{LP-CF-re} in the standard form \cite{barvinok_02}, we introduce the following spaces:

\begin{itemize}
    \item $\displaystyle E_1 \defeq \FSMeasures(\bm O_X)$.
    \item $\displaystyle F_1 \defeq C(O_X)$, the dual space of $E_1$.
    \item $\displaystyle E_2 \defeq \prod_{C \in \Mc} \FSMeasures(\bm O_C)$.
    \item $\displaystyle F_2 \defeq \prod_{C \in \Mc} C(O_C) $, the dual space of $E_2$.
\end{itemize}
The  dualities $\langle \dummy , \dummy \rangle_1 : E_1 \times F_1 \longrightarrow \R$ and $\langle \dummy , \dummy \rangle_2 : E_2 \times F_2 \longrightarrow \R$ are defined as follows:
\begin{align*}
    & \forall \, \mu \in E_1, \; f \in F_1, & & \langle \mu,f \rangle_1 \defeq \intg{O_X}{f}{\mu} \\
    & \forall \, (\nu_C) \in E_2, \; (f_C) \in F_2, & & \langle  (\nu_C),(f_C) \rangle_2 \defeq \sum_{C \in \Mc} \intg{O_C}{f_C}{\nu_C}\, ,
\end{align*}
where, for simplicity, we have omitted $C \in \Mc$ as a subscript for the families of functions.
We fix $K_1$ to be the convex cone of positive measures in $E_1 = \FSMeasures(\bm O_C)$ and $K_2$ to be the convex cone of families of positive measures in $E_2 = \prod_{C \in \Mc} \FSMeasures(\bm O_C)$. Then $K_1^*$ is the convex cone of positive function in $F_1 = C(O_X)$ and $K_2^*$ is the convex cone of families of positive functions in $F_2 = \prod_{C \in \Mc} C(O_C)$.

Let $\fdec{A}{E_1}{E_2}$ be the following linear transformation: for $\mu \in E_1$,
\begin{equation*}
 A(\mu) \defeq (\mu|_C)_{C \in \Mc} \in E_2 \Mdot
\end{equation*}
We also define the linear transformation $\fdec{A^*}{F_2}{F_1}$ as follows: for $(f_C) \in F_2$,
\begin{equation*}
    A^*((f_C)) \defeq \sum_{C \in \Mc} f_C \circ \rho^X_C \in F_1 \Mdot
\end{equation*}
We can verify that $A^*$ is the dual transformation of $A$: given $\mu \in E_1$ and $(f_C) \in F_2$, we have
\begin{align*}
\langle A(\mu), (f_C) \rangle_2 & = \langle  (\mu|_C), (f_C)  \rangle_2 \\
& = \sum_{C \in \Mc} \intg{O_C}{f_C}{\mu|_C} \\
& = \intg{O_X}{\sum_{C \in \Mc} f_C \circ \rho^X_C }{\mu} \\
& = \langle  \mu , \sum_{C \in \Mc} f_C \circ \rho^X_C  \rangle_1 \\
& = \langle  \mu, A^*((f_C))  \rangle_1 \Mdot
\end{align*}

Now fixing the vector function in the objective to be $c \defeq - \bm 1_{O_X} \in F_1$ and the vector in the constraints to be $b \defeq (-e_C)_{C \in \Mc} \in E_2$, the program \refprog{LP-CF-re} (resp. \refprog{DLP-CF-re})  can be expressed as in the standard form given in \cite{barvinok_02}:
\leqnomode
\begin{flalign*}
    \label{prog:LPstdform}
    \tag*{(LP)}
    \hspace{3cm} \left\{
    \begin{aligned}
            & \quad \text{Find } e_1 \in E_1 \\
            & \quad \text{minimising } \langle e_1,c \rangle_1 \\
            & \quad \text{subject to:}  \\
            & \hspace{1cm} \begin{aligned}
            & A(e_1) \geq_{K_2} b  \\
            & e_1 \geq_{K_1} 0 \Mdot
            \end{aligned} 
    \end{aligned}
    \right. &&
\end{flalign*}
\begin{flalign*}
    \label{prog:DLPstdform}
    \tag*{(D-LP)}
    \hspace{3cm}\left\{
    \begin{aligned}
        & \quad \text{Find } f_2 \in F_2 \\
        & \quad \text{maximising } \langle b,f_2 \rangle_2 \\
        & \quad \text{subject to:}  \\
        & \hspace{1cm} \begin{aligned}
        & A^*(f_2) \leq_{K_1^*} c \\
        & f_2 \geq_{K_2} 0 \Mdot
        \end{aligned}
    \end{aligned}
    \right. &&
\end{flalign*}
\reqnomode

\reqnomode
Note that the minus sign in the vectors $c$ and $b$ was added because we chose the primal program in the standard form to be a minimisation problem while the primal program \refprog{LP-CF-re} at hand is a maximisation problem. 

\section{Proof of Proposition \ref{prop:zero_duality}: zero duality gap}
\label{sec:appendix_proof_zeroduality}

In this appendix we give a full proof of Proposition \ref{prop:zero_duality}; \ie that strong duality holds between problems \refprog{LP-CF} and \refprog{DLP-CF}.
\begin{proof}
To show strong duality, we rely on 
\cite[Theorem 7.2]{barvinok_02}. Because $\mu_0= \bm 0_{\bm O_X}$---the measure that assigns $0$ to every measurable set of $\bm O_X$--- is a feasible solution for \refprog{LP-CF-re} and the noncontextual fraction lies between 0 and 1, \refprog{LP-CF-re} is consistent with finite value. Thus it suffices to show that the following cone
\[
    \Kc = \setdef{\left( A(\mu), \langle \mu,c \rangle_1 \right)}{\mu \in K_1} = \setdef{\left(\, (\mu|_C)_C, \mu(O_X)\, \right)}{\mu \in K_{1}} 
\]
is weakly closed in $E_2 \oplus \R$ (\ie closed in the weak topology of $K_1$) where we recall that $K_1$ is the convex cone of positive measures in $E_1 = \FSMeasures(\bm O_X)$.

We first notice that the linear transformation $A$ is a bounded linear operator and thus continuous. Boundedness comes from the fact that for all $\mu \in K_1$,
\begin{align}
\begin{split}
    \norm{A(\mu)}_{E_2} &= \norm{(\mu|_C)_C}_{E_2}  \\
    & = \sum_{C \in \Mc} \norm{\mu|_C}_{\FSMeasures(\bm O_C)}  \\
    & \leq \sum_{C \in \Mc}  \norm{\mu}_{E_1} \label{eq:boundedness}\\
    &  = \vert \Mc \vert \norm{\mu}_{E_1}  \Mcomma
    \end{split}
\end{align}
where we take the strong topology---\ie the norm induced by the total variation distance---on finite-signed measure spaces. It is defined as:
\begin{equation*}
    \norm{\mu}_{\FSMeasures(\bm U)} = \vert \mu \vert(U)\Mdot
\end{equation*}
We also equip the finite product space $E_2 = \prod_{C \in \Mc} \FSMeasures(\bm O_C)$ with the norm obtained by summing\footnotemark\ the individual total variation norms. The inequality in Eq.~(\ref{eq:boundedness}) is due to the fact that $\mu \in K_1$ so this is a positive measure and thus $\norm{\mu|_C}_{\FSMeasures(\bm O_C)} \leq \norm{\mu}_{E_1}$. This, of course, extends to the weak topology.

\footnotetext{Categorically, this is a coproduct.}
Secondly, we consider a sequence 
$(\mu^k)_{k \in \N}$ in $K_1$
and we want to show that the accumulation point $((\Theta_C)_C,\lambda) = \lim_{k \rightarrow \infty} \left( A(\mu^k), \langle \mu^k,c \rangle_1 \right)$ 
belongs to $\Kc$, where $\Theta = (\Theta_C)_C \in E_2$ and $\lambda \in \R$. 
If we consider the product of indicator functions $(\mathbf{1}_{O_C})_C \in F_2$ then 
$\langle A(\mu^k), (\mathbf{1}_{O_C}) \rangle_2 = \sum_{C \in \Mc} \mu|_C^k(O_C) \longrightarrow_{k} \sum_{C \in \Mc} \Theta_C(O_C) < \infty$ as  $\Theta_C$ is a finite measure for all maximal contexts $C \in \Mc$. Then because $\Mc$ is a covering family of $X$, $\forall k \in \N$, we have that $\mu^k(O_X) \leq \sum_{C \in \Mc} \mu^k|_C(O_C) < \infty$.
Since $(\mu^k) \in K_1^\N$ is a sequence of positive measures, this implies that $(\mu^k)$ is bounded.
Next, by weak-$*$ compactness of the unit ball (Alaoglu's theorem \cite{luenberger97}), there exists a subsequence $(\mu^{k_i})_{i}$ that converges weakly to an element $\omega \in K_1$. 
By continuity of $A$, we obtain that the accumulation point is such that $((\Theta_C)_C,\lambda) = \left( A(\omega), \langle \omega,c \rangle_1 \right) \in \Kc $.
\end{proof}


\section{The Lasserre--Parrilo hierarchy}
\label{sec:appendix_Lasserrehierarchy}

Below we introduce the Lasserre--Parrilo hierarchy for relaxing infinite-dimensional linear programs known as Generalised Moment Problems \cite{lasserre2001global,lasserre10,parrilo2003semidefinite}.
We start by giving insightful results: Subsection~\ref{subsec:sos} provides results concerning the representation of positive polynomials while Subsection~\ref{subsec:moment} provides results to understand when a sequence can be represented by a Borel measure.

\paragraph{Notation, terminology}
We work in $\R^d$ for $d \in \N^*$. We fix $\bm K$ to be a generic Borel measurable subspace of $\bm \R^d$. Below we fix some multi-index notations.
Let $\Rpolm$ denote the ring of real polynomials in the variables $\bm x \in \R^d$, and let $\Rpolkm \subset \Rpolm$ contain those polynomials of total degree at most $k$.
The latter forms a vector space of dimension $s(k) \defeq \binom{d+k}{k}$,
with a canonical basis consisting of monomials
$\bm x^{\bm \alpha} \defeq x_1^{\alpha_1}\cdots x_d^{\alpha_d}$
indexed by the set $\N^d_k \defeq \setdef{\bm \alpha \in \N^d}{\lvert \bm \alpha \rvert \leq k}$ where $\lvert\bm \alpha\rvert\defeq\sum_{i=1}^d\alpha_i$. 
For $k \in \N$, $\bm x \in \R^d$, we define $\bm \vrm_k(\bm x) \defeq (\bm x^{\bm \alpha})_{\vert \bm \alpha \vert \leq d} = (1,x_1,\dots,x_n,x_1^2,x_1x_2,\dots,x_n^k)^T $ the vector of monomials of total degree less or equal than $k$.

Any $ p \in \Rpolkm$
is associated with a vector of coefficients  $\bm p \defeq (p_{\bm \alpha}) \in \R^{s(k)}$ via expansion in the canonical basis 
as $p(\bm x) = \sum_{\bm \alpha \in \N^d_k} p_{\bm \alpha}\bm x^{\bm \alpha}$. 

\paragraph{Moment problem in probability}

Given a finite set of indices $\Gamma$, a set of reals $\enset{\gamma_j : j \in \Gamma}$ and functions $\fdec{h_j}{K}{\R}$, $j \in \Gamma$, that are integrable with respect to every measure $\mu \in \FSMeasures(\bm K)$, the corresponding \textit{Global Moment Problem} (GMP) can be expressed as:

\leqnomode
\begin{flalign*}
   \label{prog:GMPstdform}
    \tag*{(GMP)}
    \hspace{2cm}\left\{
        \begin{aligned}
            & \quad \text{Find } \mu \in\FSMeasures(\bm K) \\
            & \quad \text{maximising } \mu(\bm K) \\
            & \quad \text{subject to:}  \\
            & \hspace{1cm} \begin{aligned}
                &  \forall j \in \Gamma,\quad \intg{K}{h_j}{\mu} \leq \gamma_j \Mdot
            \end{aligned}
        \end{aligned}
    \right. &&
\end{flalign*}
\reqnomode
It dual program can be expressed as:
\leqnomode
\begin{flalign*}
   \label{prog:D-GMPstdform}
    \tag*{(D-GMP)}
    \hspace{2cm}\left\{
        \begin{aligned}
            & \quad \text{Find } \bm \lambda \in \R^\Gamma \\
            & \quad \text{minimising } \sum_{j \in \Gamma} \gamma_j \lambda_j \\
            & \quad \text{subject to:}  \\
            & \hspace{1cm} \begin{aligned}
                &  \forall \bm x \in K,\quad \sum_{j \in \Gamma} \lambda_j h_j(\bm x) - \bm x \geq 0 \\
                & \forall j \in \Gamma, \quad \lambda_j \geq 0 \Mdot
            \end{aligned}
        \end{aligned}
    \right. &&
\end{flalign*}
\reqnomode

\subsection{Positive polynomials and sum-of-squares}
\label{subsec:sos}

Here we present the link between positive polynomials and sum-of-squares representation so that we can derive a converging hierarchy of restriction problems for program~\refprog{D-GMPstdform}.

\begin{definition}
A polynomial $p \in \Rpolm$ is a \textit{sum-of-squares (SOS)} polynomial if there exists a finite family of polynomials $\family{q_i}_{i \in I}$ such that
$p = \sum_{i\in I} q_i^2$.
\end{definition}
\noindent SOS polynomials are widely used in convex optimisation.
We will denote by $\SOSm \subset \Rpolm$ the set of (multivariate) SOS polynomials, and $\SOSkm \subset \SOSm$ the set of SOS polynomials of degree at most $2k$. The following proposition hints towards the reason why it is desirable to be able to look for a sum-of-squares decomposition: it can be cast as a semidefinite optimisation problem.
\begin{proposition}[Prop. 2.1, \cite{lasserre10}]
\label{prop:sosdecomposition}
A polynomial $p \in \Rpolm_{2k}$ has a sum-of-squares decomposition if and only if there exists a real symmetric positive semidefinite matrix $Q \in \SymMatrices{s(k)}$ such that $\forall \bm x \in \R^d$, $p(\bm x) = \bm \vrm_k(\bm x)^T Q \bm \vrm_k(\bm x)$.
\end{proposition}
Then we will be looking at conditions under which a nonnegative polynomial can be expressed as a sum-of-squares polynomial. This is in essence the question raised by Hilbert in his 17$^\text{th}$ conjecture \cite{hilbert1888darstellung}. 
\begin{definition}
For a family $q = (q_j)_{j\in\enset{1, \ldots, m}}$ of polynomials, the set:
\[
    Q(q) \defeq  \setdef{ \sum_{j=0}^m \sigma_j q_j}{(\sigma_j)_{j\in\enset{0,\ldots, m}} \subset \SOSm}
\]
is a convex cone in $\Rpolm$ called the \textit{quadratic module} generated by the family $q$ with, for convenience, $q_0 = 1$ added. For $k \in \N$, we define $Q_k(q)$ to be the quadratic module $Q(q)$ where we further impose that $(\sigma_j)_{j\in\enset{0,\ldots, m}} \subset \SOSkm$ \ie we limit the degree of SOS polynomials.
\end{definition}

In the family of polynomials $(g_j)_{j\in\enset{1, \ldots, m}}$ we add $g_0 = 1$ for convenience. 

\begin{assumption} \label{ass:algebraic}
{\upshape Let $K \subset \R^d$. We make the following three assumptions on $K$.
\begin{enumerate}[label=(\roman*)]
\item \label{ass:semialgebraic}  Suppose $K$ is a basic semi-algebraic set \ie there exists a family of polynomials $g = (g_j)_{j\in\enset{1, \ldots, m}} \in \Rpolm^m$ of degrees deg($g_j$) respectively such that:
\[
    K \defeq \setdef{\bm x \in \R^d}{\forall j = 1,\dots,m, \; g_j(\bm x) \geq 0}.
\]
\item \label{ass:compact} Further suppose that $K$ is compact. 
\item \label{ass:archimidean} Finally suppose that there exists $u \in Q(q)$ such that the level set $\setdef{\bm x \in \R^d}{u(\bm x) \geq 0}$ is compact.
\end{enumerate}
}
\end{assumption}

The following theorem is the key result that we will exploit for deriving the hierarchy of SDP restrictions for the dual program~\refprog{GMPstdform}. 
\begin{theorem}[Putinar's Positivellensatz \cite{Putinar93}]
Let $K \subset \R^d$ satisfy Assumptions~\ref{ass:algebraic}. If $p \in \Rpolm$ is strictly positive on $K$ then $p \in Q(g)$, that is
\[
    p = \sum_{j=0}^m \sigma_j g_j  
\]
for some sum-of-squares polynomials $\sigma_j \in \SOSm$ for $j=0,1,\dots,m$.
\label{th:putinar}
\end{theorem}
\noindent A proof can also be found in \cite{Laurent2009}.

Using the results above and Assumption~\ref{ass:algebraic}, one can derive a hierarchy of SDPs \cite{lasserre10} which provide a converging sequence of optimal values towards the value of program~\refprog{D-GMPstdform}:

\leqnomode
\begin{flalign*}
   \label{prog:DSDP-GMP}
    \tag*{(D-GMP$^k$)}
    \hspace{2cm}\left\{
        \begin{aligned}
            & \quad \text{Find } \bm \lambda = (\lambda_j)_{j \in \Gamma} \in \R^\Gamma \text{ and } \forall j=0,\dots,m, \, f_j \in \SOSm_{k- \lceil \frac{\text{deg}(g_j) }{2}\rceil} \hspace{-3cm}\\
            & \quad \text{minimising } y_0 \\
            & \quad \text{subject to:}  \\
            & \hspace{1cm} \begin{aligned}
                &  \sum_{j \in \Gamma} \lambda_j h_j - \bm 1_{\bm K} = \sum_{j=0}^m f_j g_j  \\
                &  \forall j\in \Gamma, \quad \lambda_j \geq 0 \Mdot
            \end{aligned}
        \end{aligned}
    \right. &&
\end{flalign*}
\reqnomode

\subsection{Moment sequences and moment matrices}
\label{subsec:moment}

In this subsection, we want to understand why the program \refprog{LP-CF} can be relaxed so that a converging hierarchy of SDPs can be derived. The program \refprog{LP-CF} is essentially a maximisation problem on finite-signed Borel measures with additional constraints such as the fact that these are proper measures (\ie they are nonnegative). We will represent a measure by its moment sequence and find conditions for which this moment sequence has a (unique) representing Borel measure. 

\begin{definition}
Given a sequence $\bm y = (y_{\bm \alpha})_{\bm \alpha \in \N^d}\in \R^{\N^d}$,
we define the linear functional $\fdec{L_{\bm y}}{\Rpolm}{\R}$ by 
\[
L_{\bm y}(p) \defeq \sum_{\bm \alpha \in \N^d}p_{\bm \alpha} y_{\bm \alpha}.
\]
\label{def:RieszFunctional}
\end{definition}

\begin{definition}
Given a measure $\mu \in \Measures(\bm K)$, its \textit{moment sequence} $\bm y = (y_{\bm \alpha})_{\bm \alpha \in \N^d} \in \R^{\N^d}$ is given by
\begin{equation}
    y_{\bm \alpha} \defeq \intg{\bm K}{\bm x^{\bm \alpha}}{\mu(\bm x)} \Mdot
    \label{eq:ch0_momentseq}
\end{equation}
We say that $\bm y$ has a unique representing measure $\mu$ when there exists a unique $\mu$ such that Eq.~\eqref{eq:ch0_momentseq} holds. If $\mu$ is unique then we say it is determinate (\ie determined by its moments).
\end{definition}

\noindent The linear functional $L_{\bm y}$ then gives integration of polynomials with respect to $\mu$ \ie for any $p \in\Rpolm$:
\begin{align*}
L_{\bm y}(p) &=
\sum_{\bm \alpha \in \N^d}p_{\bm \alpha} y_{\bm \alpha}
=
\sum_{\bm \alpha \in \N^d}p_{\bm \alpha} \intg{\bm K}{\bm x^{\bm \alpha}}{\mu(\bm x)}
=
\intg{\bm K}{\sum_{\bm \alpha \in \N^d}p_{\bm \alpha}\bm x^{\bm \alpha}}{\mu(\bm x)} \\
&=
\intg{\bm K}{p(\bm x)}{\mu(\bm x)} \\
&=
\intg{\bm K}{p}{\mu},
\end{align*}
where we reversed summation and integration because the sum is finite since $p$ is a polynomial.

The following theorem is often used in optimisation theory over measures as it provides a necessary and sufficient condition for a sequence to have a representing measure. 
\begin{theorem}[Riesz-Haviland \cite{haviland1936momentum}]
Let $\bm y = (y_{\bm \alpha})_{\bm \alpha \in \N^d} \in \R^{\N^d}$ and suppose that $K \subseteq \R^d$ is closed. Then $\bm y$ has a representation (nonnegative) measure \ie there exists $\mu$ a measure on $K$ such that:
\begin{equation*}
    \forall \bm \alpha \in \N^d, \; \intg{K}{\bm x^{\bm \alpha}}{\mu} = y_{\bm \alpha}
\end{equation*}
if and only if $L_{\bm y}(p) \geq 0$ for all polynomials $p \in \Rpolm$ nonnegative on $K$.
\label{th:RieszHaviland}
\end{theorem}
\noindent We recall that for $k \in \N$, $s(k) = \binom {d+k}k$.
\begin{definition}
\label{def:momentmatrix}
For each $k \in \N$, the \textit{moment matrix} of order $k$, $M_k(\bm y) \in \SymMatrices{s(k)}$, of a truncated sequence $(y_{\bm \alpha})_{\bm \alpha \in \N^d_{2k}}$
is the $s(k) \times s(k)$ symmetric matrix with rows and columns indexed by $\N^d_k$ (\ie by the canonical basis for $\Rpolkm$) defined as follows: for any $\bm \alpha, \bm \beta \in \N^d_k$,
\[
    \left(M_k(\bm y)\right)_{\bm \alpha,\bm \beta} \defeq L_{\bm y}(\bm x^{\bm \alpha + \bm \beta}) = y_{\bm \alpha + \bm \beta} \Mdot
\]
\end{definition}

\begin{definition}
\label{def:localisingmatrix}
Given a polynomial $p \in \Rpolm$, the \textit{localising matrix} $M_k(p \bm y) \in \Matrices{s(k)}{\R}$ of a moment sequence $(y_{\bm \alpha})_{\bm \alpha \in \N^d} \in \R^{\N^d}$
is defined by: for all $\bm \alpha, \bm \beta \in \N^d_{k}$,
\[
     \left(M_k(p \bm y)\right)_{\bm \alpha, \bm \beta} \defeq  L_{\bm y}(p(x) x^{\bm \alpha + \bm \beta}) = \sum_{\gamma \in \N^d} p_{\gamma} y_{\bm \alpha + \bm \beta + \gamma} \Mdot
\]
\end{definition}

\noindent The localising matrix reduces to the moment matrix for $p=1$. For well-defined moment sequences, \ie sequences that have a representing finite Borel measure, moment matrices and localising matrices are positive semidefinite, which provides insight on the reason why problem \refprog{LP-CF} can be relaxed to a problem with positive semidefiniteness constraints. 

\begin{proposition}
Let $\bm y = (y_{\bm \alpha})_{\bm \alpha \in \N^d} \in \R^{\N^d}$ be a sequence of moments for some finite Borel measure $\mu$ on $\bm K$. Then for all $k\in \N$, $M_k(\bm y) \succeq 0 $. If $\mu$ has support contained in the set $\setdef{\bm x \in K}{g(\bm x) \geq 0}$ for some polynomial $g \in \Rpolm$ then, for all $k \in \N$, $M_k(g \bm y) \succeq 0$.
\end{proposition}

\begin{proof}
Let $\bm y = (y_{\bm \alpha})$ be the moment sequence of a given Borel measure $\mu$ on $\bm K$. Fix $k \in \N$. For any vector $\bm v \in \R^{s(k)}$ (noting that $\bm v$ is canonically associated with a polynomial $v \in \Rpolkm$ in the basis $(\bm x^{\bm \alpha})$):
\begin{align*}
	\bm v^T M_k(\bm y) \bm v &= \sum_{\bm \alpha,\bm \beta \in \N^d_k}  v_{\bm \alpha} y_{\bm \alpha + \bm \beta} v_{\bm \beta}  \\
	& = \sum_{\bm \alpha, \bm\beta \in \N^d_k} v_{\bm\alpha} v_{\bm\beta} \intg{K}{\bm x^{\bm\alpha + \bm\beta}}{\mu}  \\
	& = \intg{K}{\left( \sum_{\bm\alpha \in \N^d_k} v_{\bm\alpha} \bm x^{\bm \alpha} \right)^2}{\mu} \\
	& = \intg{K}{v^2(\bm x)}{\mu} \geq 0 \Mdot 
\end{align*} 
Thus $M_{k}(\bm y) \succeq 0 $. 

Similarly we can prove that the localising matrix $M_k(g \bm y)$ is positive semidefinite when $g$ is a nonnegative polynomial on the support of $\mu$. Indeed for all $\bm v \in \R^{s(k)}$:
\[
	\bm v^T M_k(g \bm y) \bm v = \intg{K}{v^2(\bm x) g(\bm x)}{\mu} \geq 0 \Mcomma
\]
which concludes the proof.
\end{proof}

The following theorem, which is the dual facet of Theorem~\ref{th:putinar}, is the key result for deriving the hierarchy of SDP relaxations for the primal problem \refprog{LP-CF}. It provides a necessary and sufficient condition for a sequence to have a representing measure.
\begin{theorem}[Theorem 3.8 \cite{lasserre10}]
Let $\bm y = (y_{\bm \alpha})_{\bm \alpha \in \N^d} \in \R^{\N^d}$ be a given infinite sequence in $\R$. Let $K \subset \R^d$ satisfy Assumptions~\ref{ass:algebraic}. Then $\bm y$ has a finite Borel representing measure with support contained in $K$ if and only if:
\begin{align*}
    M_k(\bm y) \succeq 0, \quad &\forall k \in \N, \\
    M_k(g_j \bm y) \succeq 0, \quad &\forall j=1,\dots,m , \; \forall k \in \N.
\end{align*}
\label{th:representingmeasure}
\end{theorem}

Using the results above and Assumption~\ref{ass:algebraic}, one can derive a hierarchy of SDPs \cite{lasserre10} which provide a converging sequence of optimal values towards the value of program~\refprog{GMPstdform}:
\leqnomode
\begin{flalign*}
   \label{prog:SDP-GMP}
    \tag*{(GMP$^k$)}
    \hspace{2cm}\left\{
        \begin{aligned}
            & \quad \text{Find } \bm y = (y_{\bm \alpha})_{\bm \alpha \in \N^d_{2k}} \in \R^{s(2k)} \\
            & \quad \text{maximising } y_0 \\
            & \quad \text{subject to:}  \\
            & \hspace{1cm} \begin{aligned}
                &  \forall j \in \Gamma,\quad L_{\bm y}(h_j) \leq \gamma_j \\
                &  M_k(\bm y) \succeq 0 \\
                &  \forall i \in 1,\dots,m, \quad M_{k - \lceil \frac{\text{deg}(g_i)}{2} \rceil }(g_i \bm y) \succeq 0 \Mdot
            \end{aligned}
        \end{aligned}
    \right. &&
\end{flalign*}
\reqnomode
We refer readers to \cite{lasserre10} for the proof of convergence of the hierarchies given by programs \refprog{SDP-GMP} and \refprog{DSDP-GMP}.

\section{\texorpdfstring{Duality between programs (SDP-CF$^k$) and (DSDP-CF$^k$)}{Duality between programs (SDP-CFk) and (DSDP-CFk)}}
\label{sec:appendix_SDP_duality}

As mentioned above, we chose to derive programs \refprog{SDPk-CF} and \refprog{DSDPk-CF} using dual arguments. These programs should therefore be dual to one another, which will immediately provide weak duality. We prove this for completeness. 

\begin{proposition}
The program  {\upshape \refprog{DSDPk-CF}} is the dual formulation of the program {\upshape \refprog{SDPk-CF}}. 
\label{prop:weakduality}
\end{proposition}

\begin{proof}
We start by rewriting $M_k(\bm y)$ as $\sum_{\bm \alpha \in \N^d_k} \bm y_{\bm \alpha} A_{\bm \alpha}$ and $M_{k-1}(g_j \bm y)$ as $\sum_{\bm \alpha \in \N^d_k} y_{\bm \alpha} B^j_{\bm \alpha}$ for $1 \leq j \leq m$ and for appropriate real symmetric matrices $A_{\bm \alpha}$ and $(B^j_{\bm \alpha})_j$. For instance, in the basis $(\bm x^{\bm \alpha})$:
\[
   ( A_{\bm \alpha} )_{\bm s, \bm t} = \bigg( 
    \begin{cases}
        1 \text{ if } \bm s+\bm t = \bm \alpha \\
        0 \text{ otherwise}
    \end{cases}
    \bigg)_{\bm s, \bm t} \Mdot
\]
From $A_{\bm \alpha}$, we also extract $A_{\bm \alpha}^C$ for $C \in \Mc$ in order to rewrite $M_k(\bm y|_C)$ as $\sum_{\alpha \in \N^d_k} \bm y_{\bm \alpha} A_{\bm \alpha}^C$. This amounts to identifying which matrices $(A_{\bm \alpha})$ contribute to a given context $C \in \Mc$. Then \refprog{SDPk-CF} can be rewritten as:
\leqnomode
\begin{flalign*}
    \tag*{(SDP-CF$^{k}$)}
    \hspace{3cm}\left\{
    \begin{aligned}
        & \quad \text{Find } \bm y \in \R^{s(2k)} \\
        & \quad \text{maximising } y_{\bm 0} \\
        & \quad \text{subject to:}  \\
        &  \hspace{1cm} \begin{aligned}
            & \forall C \in \Mc,\; M_k(\bm y^{e,C}) - \sum_{\bm \alpha \in \N^d_k} y_{\bm \alpha} A_{\bm \alpha}^C \succeq 0, \\
            & \forall j=1,\dots,m,\; \sum_{\bm \alpha \in \N^d_k} y_{\bm \alpha} B^j_{\bm \alpha} \succeq 0, \\
            & \sum_{\bm \alpha \in \N^d_k} y_{\bm \alpha} A_{\bm \alpha} \succeq 0 \Mdot
        \end{aligned}
    \end{aligned}
    \right. &&
\end{flalign*}
\reqnomode
Via the Lagrangian method, this is equivalent to:
\[
	\sup_{\bm y \in \R^{s(k)}} \; \inf_{ \substack{(X^C), (Y_j), Z \\ \text{ SDP matrices}}} \; \Lc(\bm y,(X_C),Y,(Z_j)),
\]
with
\begin{align*}
    \Lc(\bm y,(X^C),(Y_j),Z) &= y_{\bm 0} \\
    & + \sum_{C \in \Mc} \Tr(M_k(\bm y^{e,C}) X^C) - \sum_{C \in \Mc} \sum_{\bm \alpha \in \N^d_k} y_{\bm \alpha} \Tr(A_{\bm \alpha}^C X^C) \\
    & + \sum_{j=1}^m \sum_{\bm \alpha \in \N^d_k} y_{\bm \alpha} \Tr(B_{\bm \alpha}^j Y_j) \\
    & + \sum_{\bm \alpha \in \N^d_k} y_{\bm \alpha} \Tr(A_{\bm \alpha} Z) \Mdot
\end{align*}
The dual program corresponds to 
\[
	\inf_{ \substack{(X^C), (Y_j), Z \\ \text{ SDP matrices}}} \; \sup_{\bm y \in \R^{s(k)}} \; \Lc(\bm y,(X^C),(Y_j),Z) \Mdot
\]
We rewrite the Lagrangian as:
\begin{align*}
    \Lc(\bm y,(X_C),Y,(Z_j)) &= \sum_{C \in \Mc} \Tr(M_k(\bm y^{e,C}) X^C) \\
    & + \sum_{\bm \alpha \in \N^d_k} y_{\bm \alpha} \left( \delta_{\bm \alpha, \bm 0}  - \sum_{C \in \Mc} \Tr(A_{\bm \alpha}^C X^C) + \sum_{j=1}^m \Tr(B_{\bm \alpha}^j Y_j) + \Tr(A_{\bm \alpha} Z) \right) \Mdot
\end{align*}
Then the dual program of \refprog{SDPk-CF} reads:
\leqnomode
\begin{flalign*}
    \hspace{2cm}\left\{
    \begin{aligned}
        & \quad \text{Find } \family{X^C}_{C \in \Mc}, \family{Y_j}_{j=1,\dots,m} \text{ and } Z \text{ SDP matrices } \\
        & \quad \text{maximising } \sum_{C \in \Mc} \Tr(M_k(\bm y^{e,C}) X^C) \\
        & \quad \text{subject to:}  \\        
        & \hspace{1cm} \begin{aligned}
            &  \forall \bm \alpha \in \N^d_k, \; \sum_{C \in \Mc} \Tr(A_{\bm \alpha}^C X^C) - \sum_{j=1}^m \Tr(B_{\bm \alpha}^j Y_j) - \Tr(A_{\bm \alpha} Z) = \delta_{\bm \alpha, \bm 0} \Mdot
        \end{aligned}
    \end{aligned}
    \right. &&
\end{flalign*}
\reqnomode
We finally show that the above program exactly corresponds to \refprog{DSDPk-CF}. We start with the objective. For all $C \in \Mc$, with $X^C$ a positive semidefinite matrix:
\begin{align*}
    \Tr(M_k(\bm y^{e,C}) X^C) & = \sum_{\bm \alpha} \sum_{\bm \beta} (M_k(\bm y^{e,C}))_{\bm \alpha \bm \beta} (X^C)_{\bm \beta \bm \alpha} \\
    & = \sum_{\bm \alpha} \sum_{\bm \beta} \bm y^{e,C}_{\bm \alpha + \bm \beta} (X^C)_{\bm \alpha \bm \beta} \\
    & = \sum_{\bm \alpha} \sum_{\bm \beta} \intg{\Oc_C}{\bm x^{\bm \alpha + \bm \beta}}{e_C} (X^C)_{\bm \alpha \bm \beta} \\
    & = \intg{\Oc_C}{ \bm \vrm_k(\bm x)^T X^C \bm \vrm_k(\bm x)}{e_C} \\
    & = \intg{\Oc_C}{ p_C }{e_C},
\end{align*}
with for all $C \in \Mc$, $p_C \in \SOSkm$ a sum-of-squares polynomial via Proposition~\ref{prop:sosdecomposition} and where we used $\bm \vrm_k(\bm x)$ the vectors of monomials of maximal total degree $k$.

Now, to retrieve the constraint, we multiply each side by $\bm x^{\bm \alpha}$ and we sum for all $\bm \alpha$:
\[
    \sum_{C \in \Mc} \Tr(\sum_{\bm \alpha} \bm x^{\bm \alpha} A^C_{\bm \alpha} X^C) - \bm 1 = \sum_{j=1}^m \Tr(\sum_{\bm \alpha} \bm x^{\bm \alpha} B_{\bm \alpha}^j Y_j) + \Tr(\sum_{\bm \alpha} \bm x^{\bm \alpha} A_{\bm \alpha} Z)
\]

Recalling the definition of moment and localising matrices:
\begin{align*}
    & \sum_{\bm \alpha} A^C_{\bm \alpha} \bm x^{\bm \alpha} = \bm \vrm_k(\bm x) \bm \vrm_k(\bm x)^T \\
    & \sum_{\bm \alpha} B_{\bm \alpha}^j \vrm_k(\bm x) \vrm_k(\bm x)^T = g_j(\bm x) \bm \vrm_{k-1}(\bm x) \bm \vrm_{k-1}(\bm x)^T, \; \forall j=1,\dots,m 
\end{align*}

Thus, by Proposition~\ref{prop:sosdecomposition}, for appropriate sum-of-squares polynomials $(\sigma_j)_{j=0,1,\dots,m} \subset \Rpolm_{k-1}$:
\begin{align*}
    & \Tr(\sum_{\bm \alpha} \bm x^{\bm \alpha} A^C_{\bm \alpha} X^C) = p_C \circ \rho_C^X (\bm x) \\
    & \Tr(\sum_{\bm \alpha} \bm x^{\bm \alpha} B_{\bm \alpha}^j Y_j) = g_j(\bm x) \sigma_j(\bm x) \\
    & \Tr(\sum_{\bm \alpha} \bm x^{\bm \alpha} A_{\bm \alpha} Z) = \sigma_0(\bm x)
\end{align*}
This is exactly the constraint in \refprog{DSDPk-CF} with, for convenience, $g_0 = 1$.
\end{proof}
\vfill

%
%

\end{document}